\documentclass[11pt]{article}

\usepackage{a4wide}
\usepackage{parskip}
\usepackage[utf8]{inputenc}
\usepackage{amsmath}
\usepackage{amssymb}
\usepackage{xspace}
\usepackage{tikz}
\usepackage{algorithm}
\usepackage[noend]{algpseudocode}

\usepackage{todonotes}
\usepackage{calc,ifthen}
\newcounter{hours}
\newcounter{minutes}
\newcommand{\Printtime}{\setcounter{hours}{\time/60}%
\setcounter{minutes}{\time-\value{hours}*60}%
\thehours:%
\ifthenelse{\value{minutes}<10}{0}{}\theminutes}
\usepackage[nodate]{datetime} 

\usepackage{todonotes}

\usetikzlibrary{positioning,calc,shapes,arrows}
\usetikzlibrary{backgrounds}
\usetikzlibrary{matrix,shadows,arrows} 
\usetikzlibrary{decorations.pathreplacing}

\newtheorem{xdefinition}{Definition}[section]
\newtheorem{xobservation}[xdefinition]{Observation}
\newtheorem{xremark}[xdefinition]{Remark}
\newtheorem{xtheorem}[xdefinition]{Theorem}
\newtheorem{xlemma}[xdefinition]{Lemma}
\newtheorem{xproposition}[xdefinition]{Proposition}
\newtheorem{xcorollary}[xdefinition]{Corollary}
\newenvironment{definition}{\begin{xdefinition}\rm}%
{\hspace*{\fill}\raisebox{-1pt}{\boldmath$\Box$}\end{xdefinition}}
{\hspace*{\fill}\raisebox{-1pt}{\boldmath$\Box$}\end{xobservation}}
\newenvironment{remark}{\begin{xremark}\rm}%
{\hspace*{\fill}\raisebox{-1pt}{\boldmath$\Box$}\end{xremark}}
\newenvironment{theorem}{\begin{xtheorem}\rm}{\end{xtheorem}}

\newenvironment{corollary}{\begin{xcorollary}\rm}{\end{xcorollary}}
\newenvironment{proof}{\begin{trivlist}\item[]{\bf Proof }}%
{\hspace*{\fill}\raisebox{-1pt}{\boldmath$\Box$}\end{trivlist}}

\DeclareMathOperator{\OPT}{OPT}
\DeclareMathOperator{\ALG}{ALG}

\DeclareMathOperator{\BAD}{BAD}

\renewcommand{\epsilon}{\varepsilon}

\newcommand{\SIZE}[1]{\left|#1\right|}
\newcommand{\SEQ}[1]{\left\langle#1\right\rangle}
\newcommand{\SET}[1]{\left\{#1\right\}}

\newcommand{\EPSINT}{\ensuremath{(0,\frac12]}}

\newcommand{\init}{\ensuremath{\operatorname{\textit{init}}}\xspace}
\newcommand{\isempty}{\ensuremath{\operatorname{\textit{isempty}}}\xspace}
\newcommand{\findmax}{\ensuremath{\operatorname{\textit{findmax}}}\xspace}
\newcommand{\deletemax}{\ensuremath{\operatorname{\textit{deletemax}}}\xspace}
\newcommand{\pairs}{\ensuremath{\operatorname{\textit{pairs}}}\xspace}

\title{Advice Complexity of Priority Algorithms}
\author{Allan Borodin\thanks{Research is supported by NSERC.} \\ University of Toronto \\ \textsf{bor@cs.toronto.edu} \and 
Joan Boyar\thanks{Supported in part by the Independent Research Fund Denmark, Natural Sciences, grant DFF-7014-00041.} \\ University of Southern Denmark \\ \textsf{joan@imada.sdu.dk} \and 
Kim S. Larsen\footnotemark[3] \\ University of Southern Denmark \\ \textsf{kslarsen@imada.sdu.dk} \and 
Denis Pankratov\footnotemark[2] \\ Concordia University \\ \textsf{denis.pankratov@concordia.ca}}
\date{}

\begin{document}

\maketitle

\begin{abstract}
  The priority model of ``greedy-like'' algorithms was introduced by
  Borodin, Nielsen, and Rackoff in 2002. We augment this model by
  allowing priority algorithms to have access to advice, i.e., side
  information precomputed by an all-powerful oracle. Obtaining lower
  bounds in the priority model without advice can be challenging and
  may involve intricate adversary arguments.
  Since the priority model with advice is even more
  powerful, obtaining lower bounds presents additional
  difficulties. We sidestep these difficulties by developing a general
  framework of reductions which makes lower bound proofs relatively
  straightforward and routine. We start by introducing the Pair
  Matching problem, for which we are able to prove strong lower bounds
  in the priority model with advice.
  We develop a template for constructing a reduction from Pair Matching
  to other problems in the priority model with advice -- this part is
  technically challenging since the reduction needs to define a valid
  priority function for Pair Matching while respecting the priority function
  for the other problem.
  Finally, we apply the template to obtain lower bounds for a
  number of standard discrete optimization problems.
\end{abstract}


\section{Introduction}
\label{sec:intro}
Greedy algorithms are among the first class of algorithms studied in an undergraduate computer science curriculum. They are among the simplest and fastest algorithms for a given optimization problem, often achieving a reasonably good approximation ratio, even when the problem is NP-hard. In spite of their importance, the notion of a greedy algorithm is not well defined. This might be satisfactory for studying upper bounds; when an algorithm is suggested, it does not matter much whether everyone agrees that it is greedy or not. However, lower bounds (inapproximation results) require a precise definition.
Perhaps giving a precise definition for all greedy algorithms is not possible, 
since one can provide examples 
that seem to be outside the scope of the given model.  

Setting this philosophical question aside, we follow the model of greedy-like algorithms due to Borodin, Nielsen, and Rackoff~\cite{BorodinNR2003}. The {\it fixed priority model} captures the observation that many greedy algorithms work by first sorting the input items according to some priority function, and then, during  a single pass over the sorted input, making online irrevocable decisions for each input item. This model is similar to the online algorithm model with an additional preprocessing step of sorting inputs. Of course, if any sorting function is allowed, this would trivialize the model for most applications. Instead, a total ordering on the universe of all possible input items is specified before any
input is seen, and the sorting is done according to this ordering,
after which the algorithm proceeds as an online algorithm.
This model has been adopted with respect to a broad array of topics~\cite{Regev02, AngelopoulosB2004, LeshM06, DavisI2009, Poloczek11, BorodinBLM2010, BorodinL16, BesserP17}.
In spite of the appeal of the model, there are relatively few lower bounds in this model. There does not seem to be a general method for proving lower bounds; that is, the adversary arguments tend to be ad-hoc. In addition, the basic priority model does not capture the notion of side information. The assumption that an algorithm does not know anything about the input is quite pessimistic in practice.
This issue has been addressed recently in the area of online algorithms
by considering models with advice (see~\cite{BoyarFKLM2016} for an overview). In these
models, side information, such as the number of input items or a
maximum weight of an item, is computed by an all powerful oracle and
is available to an algorithm before seeing any of the input. This
information is then used to make better online decisions. The goal is
to study trade-offs between advice length and the competitive ratio.

We introduce a general technique for
establishing lower bounds on priority algorithms with advice. These
algorithms are a simultaneous generalization of priority algorithms
and online algorithms with advice. Our technique is inspired by the
recent success of the binary string guessing problem and reductions in
the area of online algorithms with advice. We identify a difficult problem (Pair
Matching) that can be thought of as a sorting-resistant version of the
binary string guessing problem. Then, we describe the template of
gadget reductions from Pair Matching to other problems in the world of
priority algorithms with advice. This part turns out to be 
challenging, mostly because one has to ensure that priorities are
respected by the reduction. We then apply the template to a number of
classic optimization problems.
We restrict our attention to the fixed priority model.
We also note that we consider deterministic algorithms unless otherwise
specified.

{\bf Related model.}
Fixed priority algorithms with advice can be viewed in terms of the fixed priority backtracking model of 
Alekhnovich et al~\cite{alekhnovich2005toward}. 
That model starts
by ordering the inputs using a fixed priority function and then executes 
a computation tree where different decisions can be tried for the same
input item by branching in the tree, and then choosing the best result.
The lower bound results 
generally consider
how much width (maximum number of nodes for any fixed depth in the
tree) is necessary to obtain optimality 
where the
width proven is often of the form $2^{\Omega(n)}$.
In contrast, our results give a parameterized trade-off between the number
of advice bits and the competitive ratio.
However, given an algorithm
in the fixed priority backtracking model, the logarithm
of the width gives an upper bound on the number of bits
of advice needed for the same approximation ratio. Similarly,
a lower bound on the advice complexity gives a lower bound on
width.

{\bf Organization.}
We give a formal description of the models in Section~\ref{sec:prelim}. We
motivate the study of the priority model with advice in
Section~\ref{sec:motivation}. We
introduce and analyze the Pair Matching problem in Section~\ref{sec:pair_matching}. We describe the reduction framework for
obtaining lower bounds in Section~\ref{sec:redux_temp} and apply it to classic problems in Section~\ref{sec:classical_redux}. We conclude in
Section~\ref{sec:conclusion}. 


\section{Preliminaries} 
\label{sec:prelim}
We consider optimization problems for which 
we are given an objective function to minimize or maximize,
and measure our success relative to an optimal offline algorithm.

{\bf Online Algorithms with Advice.}
In an online setting, the input is revealed one item at a time by an adversary.
An algorithm makes an irrevocable
decision about the current item before the next item is revealed.
For more background on online algorithms, we refer the
reader to the texts by Borodin and El-Yaniv~\cite{BE98} and
Komm~\cite{Komm16}.

The assumption that an online algorithm does not know anything about
the input is quite often too pessimistic in practice. Depending on
the application domain, the algorithm designer may have access to
knowledge about the number of input items, the largest weight of an
input item, some partial solution based on historical
data, etc. The advice tape model for online algorithms captures the notion
of side information in a purely information-theoretic way as
follows.
An all-powerful oracle that sees the entire input prepares
the infinite advice tape with bits, which are
available to the algorithm during the entire process.
The oracle and the algorithm work in a cooperative mode
-- the oracle knows how the algorithm will use the bits
and is trying to maximize the usefulness of the
advice with regards to optimizing the given objective function.
The advice complexity of an algorithm is a function of the input length
and is the number of bits read by the algorithm in the worst case for inputs
of a given size.
For more background on online algorithms with advice, see the survey by
Boyar et al.~\cite{BoyarFKLM2016}.

{\bf Fixed Priority Model with Advice.}
Fixed priority algorithms can be formulated as
follows. Let $\mathcal{U}$ be a universe of all possible input items. An input to
the problem consists of a finite set of items
$\mathcal{I} \subset \mathcal{U}$ satisfying some consistency
conditions. The algorithm specifies a total order on $\mathcal{U}$
before seeing the input. Then, the input items are revealed
according to the total order specified by the
algorithm. The algorithm makes irrevocable decisions about the items
as they arrive.\footnote{In the adaptive priority model, the algorithm
  is allowed to specify a new ordering depending on previous items and
  decisions before a new input item is presented.}
The overall set of decisions is then evaluated according to some objective
function. The performance of the algorithm is measured by the
asymptotic approximation ratio with respect to the value provided  
by an optimal offline algorithm. The
notion of advice is added to the model as follows.  After the
algorithm has chosen a total order on $\mathcal{U}$, an all-powerful
oracle that has access to the entire input $\mathcal{I}$ creates a
tape of infinitely many bits.  The algorithm knows how the advice bits
are created and has access to them during the online decision phase.
Our interest is in how many bits of advice the algorithm uses compared
with the result it obtains.

We consider only countable universes $\mathcal{U}$. In this case,
having a total order on elements in $\mathcal{U}$ is equivalent (via a
simple inductive argument) to having a priority function
$P: \mathcal{U} \rightarrow \mathbb{R}$. The assumption of the
universe being countable is natural, but also necessary for the above
equivalence: there are uncountably many totally ordered sets that do
not embed into the reals with the standard order.

\begin{definition}
  Let $\mathcal{U}$ be the universe of input items and let
  $P: \mathcal{U} \rightarrow \mathbb{R}$ be a priority function.
  For $u_1, u_2 \in \mathcal{U}$, we write $u_1 <_P u_2$ to mean
  $P(u_1) < P(u_2)$. 
We will say that larger priority 
means that the item appears earlier in the input, i.e.,
  $u_1 <_P u_2$ means that $u_2$ appears \emph{before} $u_1$ when the
  input is given according to $P$.
\end{definition}

\emph{Example.} Kruskal's optimal algorithm for the minimum spanning tree
problem is a fixed priority algorithm without advice. The universe of
items is
$\mathcal{U} = \mathbb{N} \times \mathbb{N} \times \mathbb{Q}$. An
item $(i,j,w) \in \mathcal{U}$ represents an edge between a vertex $i$
and a vertex $j$ of weight $w$. The consistency condition on the input
is that the edge $\{i,j\}$ can be present at most once in the
input. The total order on the universe is specified by all items of
smaller weight having higher priority than all items of larger
weight, breaking ties, say, by lexicographic order on the names of
vertices. Kruskal's algorithm processes input items 
in the given order and greedily accepts those items 
that do not result in
cycles. 

In this paper, we shall only consider
the following input model for graph problems in the priority
setting:

\emph{Vertex arrival, vertex adjacency}: an input item consists of a name of a vertex together with a set of names of adjacent vertices. There is a consistency condition on the entire input: if $u$ appears as a neighbor of $v$, then $v$ must appear as a neighbor of $u$.

{\bf Binary String Guessing Problem.}
Later we introduce the Pair Matching problem that can be viewed as a priority model analogue of the following online  binary string guessing problem.

\begin{definition}
The Binary String Guessing Problem \cite{BHKKSS14} with known history (2-SGKH) is the following online problem. The input consists of $(n, \sigma=(x_1, \ldots, x_n))$, where $x_i \in \{0,1\}$. Upon seeing $x_1, \ldots, x_{i-1}$ an algorithm guesses the value of $x_i$. The actual value of $x_i$ is revealed after the guess. The goal is to maximize the number of correct guesses.
\end{definition}

B\"{o}ckenhauer et al.~\cite{BHKKSS14} provide a trade-off between the number of advice bits and the approximation ratio for the binary string guessing problem. 


\begin{theorem}[B\"{o}ckenhauer et al.~\cite{BHKKSS14}]
\label{thm:2sgkh-lb}
For the $2$-SGKH problem and any $\epsilon\in\EPSINT$,
no online algorithm reading fewer than $(1-H(\epsilon))n$ advice bits can
make fewer than $\epsilon n$ mistakes for large enough~$n$, where $H(p) = H(1-p) = - p \log (p) - (1-p) \log (1-p)$ is the binary entropy function.
\end{theorem}

{\bf Competitive and Approximation Ratios.} The performance of
online algorithms is measured by their
competitive ratios.
For a minimization problem,
an online algorithm $\ALG$ is said to be
\emph{$c$-competitive}
if there exists a constant $\alpha$ such that for all input sequences
$I$ we have $\ALG(I) \le c \OPT(I) + \alpha$, where $\ALG(I)$
denotes the cost of the algorithm on $I$ and
$\OPT(I)$ is the value achieved by an offline optimal algorithm.
The infimum of all $c$ such that $\ALG$ is $c$-competitive is $\ALG$'s
\emph{competitive ratio}.
For a maximization problem, $\ALG(I)$ is referred to as profit,
and we require that $\OPT(I) \le c \ALG(I) + \alpha$.
In this way, we always have $c \ge 1$
and the closer $c$ is to $1$, the better.
Priority algorithms are thought of as approximation algorithms
and the term (asymptotic) approximation ratio is used
(but the definition is the same).


\section{Motivation}
\label{sec:motivation}

In this section we present a motivating example for studying the priority model with advice. We present a problem that is difficult in the pure priority setting or in the online setting with advice, but easy in the priority model with advice. Furthermore, the advice is easily computed by an offline algorithm. 

The problem of interest is called Greater Than Mean (GTM). In the GTM problem, the input is a sequence $x_1, \ldots, x_n$ of rational numbers. Let $m = \sum_i x_i /n$ denote the sample mean of the sequence. The goal of an algorithm is to decide for each $x_i$ whether $x_i$ is greater than the mean or not, answering $1$ or $0$, respectively. We can also assume that the length of the sequence $n$ is known to the algorithm in advance. We start by noting that there is a trivial optimal priority algorithm with little advice for this problem.

\begin{theorem}
For Greater Than Mean, there exists a fixed priority algorithm reading at most $\lceil \log n \rceil$ advice bits, solving the problem optimally.
\end{theorem}

\begin{proof}
The priority order is such that $x_1 \geq x_2 \ldots \geq x_n$. Thus, the integers arrive in the order from largest to smallest. The advice specifies the earliest index $i \in [n]$ such that $x_i \leq  m$.
\end{proof}

Next, we show that a priority algorithm without advice has to make many errors\footnote{In Theorem~\ref{thm:GTM_lb} and  in all of our lower bound advice results, we state the result so as to include $\epsilon = \frac{1}{2}$, in which case the conditions ``fewer than $(1/2-\epsilon)$''  and ``fewer than $(1-H(\epsilon))$'' make the statements vacuously true.}.

\begin{theorem}
\label{thm:GTM_lb}
For Greater Than Mean and any $\epsilon\in\EPSINT$,
no fixed priority algorithm without advice can make fewer than $(1/2 - \epsilon)n$ mistakes
for large enough $n$.
\end{theorem}

\begin{proof}
Let $A$ be a fixed priority algorithm without advice for the GTM problem. Let $P$ be the corresponding priority function.  For simplicity, we assume that repeated items must occur consecutively when ordered according to $P$. We show how to get rid of the consecutive repeated items assumption in the remark immediately following this proof. Consider integers in the interval $[0,2]$. One of the following two cases must occur:

Case 1: there exists $i, j \in [0, 2]$ such that $i < j$ and $j >_P i$. Consider the behavior of the algorithm on the input where $j$ is presented $n-1$ times first. If the algorithm answers $1$ on the majority of these $n-1$ requests, then the last element is set to $j$, ensuring that all the $1$ answers were incorrect. If the algorithm answers $0$ on the majority, then the last element is set to $i$, ensuring that all the $0$ answers were incorrect. In either case, the algorithm makes at least $(n-1)/2$ mistakes.

Case 2: the priority function on the interval $[0, 2]$ is $0 >_P 1 >_P 2$.
Consider the behavior of the algorithm on the input where the first item is $0$ and the following
$n - 2$ items are set to $1$. If an algorithm answers $1$ on the majority of the $n-2$ items, 
then the last item is $2$. Thus, the mean is $1$,
ensuring that all the $1$ answers on the items with value $1$ are incorrect. If an algorithm answers $0$ on 
the majority of the $n-2$ items, then the last item is $1$.
Thus, the mean is strictly smaller than $1$, ensuring that all 
the $0$ answers of the algorithm on the $1$ items are incorrect. 
In either case, the 
algorithm can be made to produce errors on $(n-2)/2$ items, which is at least $(1/2 - \epsilon)n$
for $n\geq 1/\epsilon$.
\end{proof}

\begin{remark}
Suppose that we allow repeated input items to appear non-consecutively when ordered according to $P$. Formally, this can be modeled by the universe $\mathbb{Q} \times \mathbb{N}$. The input item $(x, \mathit{id})$ consists of a rational number $x$, called the value of an item, and its identification number $\mathit{id}$. Input to the GTM problem is a subset of $\mathbb{Q} \times \mathbb{N}$. The GTM problem is defined entirely in terms of values of input items, and repeated values are distinguished by their $\mathit{id}$. Fix a priority function $P$ and choose $n$ different items of value $1$, i.e., $i_1, \ldots, i_n$. Suppose that we have an item of value $0$ that is of higher priority than any of the $i_j$ and an item of value $2$ that is of lower priority than any of the $i_j$. Then we can repeat the argument of Case 2 from the proof above. 

Otherwise, pick $2n+1$ distinct items of value $1$. Call them $i_1, i_2, \ldots, i_{2n+1}$ in the decreasing order of priorities. For items $i_{n+1}, \ldots, i_{2n}$ either (a) there is no item of value $0$ of higher priority than all of them, or (b) there is no item of value $2$ of lower priority than all of them (otherwise, it is covered by the previous case). To handle (a), pick an arbitrary item of value $0$. This item has lower priority than $i_{n+1}$, and, in particular, lower priority than all of $i_1, \ldots, i_n$. This can be handled similarly to Case 1 in the proof above. Thus, the only scenario left is (b) when there is no item of value $2$ of lower priority than all of $i_{n+1}, \ldots, i_{2n}$. Pick $n$ arbitrary items of value $2$ -- they all have priority higher than $i_{2n+1}$. Thus, this can again be handled similarly to Case~1 in the proof above.
\end{remark}

Finally, we show that an online algorithm requires a lot of advice to achieve good performance for the GTM problem. The proof is a minor modification of a reduction from 2-SGKH to the Binary Separation Problem (see \cite{BKLL16} for details). We present the proof in its entirety for completeness.

\begin{theorem}
For the Greater Than Mean problem and any $\epsilon\in\EPSINT$,
no online algorithm reading fewer than $(1-H(\epsilon))(n-1)$ advice bits can
make fewer than $\epsilon n$ mistakes for large enough~$n$.
\end{theorem}

\begin{proof}
We present a reduction from the 2-SGKH problem to the GTM problem. Let $A$ be an online algorithm with advice for the GTM problem. Our reduction is presented in Algorithm~\ref{alg:2sgkh-to-gtm}. In the course of the reduction, an online input $x_1, \ldots, x_{n}$ of length $n$ for the 2-SGKH problem is converted into an online input $y_1, \ldots, y_{n+1}$ of length $n+1$ for the GTM problem with the following properties: The number of advice bits is preserved and for each $i \in [n]$, our algorithm $A$ for 2-SGKH makes a mistake on $x_i$ if and only if $A$ makes a mistake on $y_i$. This would finish the proof of the theorem. 

 Let $S = \{i \in [n] \mid x_i = 1\}$ and $T = [n]\setminus S$. The reduction uses a technique similar to binary search to make sure that $\forall i \in S$ and $\forall j \in T$ we have $y_i > y_j$, i.e., all the $y_i$ corresponding to $x_i = 1$ are larger than all the $y_j$ corresponding to $x_j = 0$. Then $y_{n+1}$ is chosen to make sure that the mean of the entire stream $y_1, \ldots, y_{n+1}$ lies between the smallest $y_i$ with $i \in S$ and the largest $y_j$ with $j \in T$. This implies that $y_i$ is greater than the mean if and only if the corresponding $x_i = 1$.

\begin{algorithm}[H]
\caption{Reduction from 2-SGKH to GTM}\label{alg:2sgkh-to-gtm}
\begin{algorithmic}
\Procedure{Reduction-2-SGKH-to-GTM}{}
\State  $\ell_1 \gets 0$, $u_1 \gets 1$
\For {$i = 1$ to $n$ }
\State $y_i \gets (\ell_i+u_i)/2$
\If {$A$ predicts $y_i$ is greater than mean}
\State predict $x_i = 1$
\Else
\State predict $x_i = 0$
\EndIf
\State {receive actual $x_i$}
\If {actual $x_i = 1$}
\State $u_{i+1} \gets y_i$, $\ell_{i+1} \gets \ell_i$
\Else
\State $u_{i+1} \gets u_i$, $\ell_{i+1} \gets y_i$
\EndIf
\EndFor
\State $y_{n+1} \gets \frac{n+1}{2}(\ell_{n+1}+u_{n+1})-\sum_{i=1}^{n} y_i$
\EndProcedure
\end{algorithmic}
\end{algorithm}

The following invariants are easy to see and are left to the reader: (1) $u_i > \ell_i$; (2) if $x_i = 1$, then $u_i > y_i \ge u_{i+1}$; (3) if $x_i = 0$, then $\ell_i < y_i \le \ell_{i+1}$.

The required  properties of the reduction follow immediately from the invariants. Let $i \in S$ and $j \in T$. Then,
$y_i \ge u_{n+1} > \ell_{n+1} \ge y_j$.
Finally, observe that $y_{n+1}$ is chosen so that the mean is 
$\sum_{i = 1}^{n+1} y_i /(n+1) = \sum_{i=1}^{n} y_i/(n+1) + y_{n+1}/(n+1) 
= (1/2)(\ell_{n+1} + u_{n+1}).$
This mean correctly separates $S$ from $T$.
\end{proof}


\section{Pair Matching Problem}
\label{sec:pair_matching}
We introduce an online problem called Pair Matching. The input consists of a sequence of $n$ distinct rational numbers between 0 and 1, i.e.,  $x_1, \ldots, x_n \in \mathbb{Q} \cap [0,1]$. After the arrival of $x_i$, an algorithm has to answer if there is a $j \in [n] \setminus \{i\}$ such that $x_i + x_j = 1$, in which case we refer to $x_i$ and $x_j$ as forming a pair and say that $x_i$ has a matching value, $x_j$. The answer ``accept'' is correct if $x_j$ exists, and ``reject'' is correct if it does not. Note that since the $x_i$ are all distinct, if $x_i=\frac12$, the correct answer is ``reject'', since $\frac12$ cannot have a matching value.

We let $\pairs(x_1, \ldots, x_n)$ denote the number of pairs in the input $x_1, \ldots, x_n$. 

\subsection{Online Setting}

Analyzing Pair Matching in the online setting is relatively straightforward for both deterministic and randomized algorithms. 

We start with a simple upper bound achieved by a deterministic online algorithm.

\begin{theorem}
For Pair Matching, there exists a $2$-competitive algorithm, answering correctly
on $n-\pairs(x_1, \ldots, x_n)$ input items.
\end{theorem}

\begin{proof}
The algorithm works as follows: suppose the algorithm has already given answers for items $x_1, \ldots, x_{i-1}$, and a new item $x_i$ arrives. If there is a $j \in [i-1]$ such that $x_i + x_j = 1$, then the algorithm answers ``accept''. Otherwise, the algorithm answers ``reject''. Observe that the algorithm always answers correctly on all items that do not come from pairs. There are $n-2 \cdot \pairs(x_1, \ldots, x_n)$ such items. Moreover, it always answers correctly on exactly a half of all items that form pairs -- namely, it answers incorrectly on the first item from a given pair and answers correctly on the second item from the given pair. Thus, the algorithm gives $\pairs(x_1, \ldots, x_n)$ correct answers in addition to the $n- 2 \cdot \pairs(x_1, \ldots, x_n)$ answers given correctly on items not forming pairs. The total number of correct answers is $n - \pairs(x_1, \ldots, x_n)$. Observe that $\pairs(x_1, \ldots, x_n) \le n/2$. Thus, this simple online algorithm gives correct answers on $\ge n/2$ items, achieving competitive ratio of at least $2$.
\end{proof}

Next, we show that the above upper bound is actually tight.

\begin{theorem}
For Pair Matching, no deterministic online algorithm can achieve a competitive ratio less than $2$.
\end{theorem}

\begin{proof}
Let $A$ be a hypothetical deterministic algorithm for Pair Matching. An adversary keeps track of the current pool of possible inputs $X$. Initially, $X = \mathbb{Q} \cap [0,1]$. An adversary picks an arbitrary number $x \in X$ as the first input item. Depending on how $A$ answers on $x$ there are two cases.

Case~1: If $A$ answers ``reject'' on $x$, then the adversary picks $1-x$ as the next input item. One can assume that $A$ answers correctly on $1-x$. Then, the adversary removes $x$ and $1-x$ from $X$ and proceeds.

Case~2: If $A$ answers ``accept'' on $x$, then the adversary removes $x$ and $1-x$ from $X$ (thus, the
matching value $1-x$ is never given) and proceeds.

Observe that in Case~1 the algorithm makes mistakes on $1/2$ of the sub-input corresponding to that case. In Case~2, removing $x$ and $1-x$ from $X$ ensures that $x$ is not part of a pair in the input. Thus, the algorithm makes mistakes on the entire sub-input corresponding to Case~2.
\end{proof}

Next, we analyze randomized online algorithms for Pair Matching. A modification of the simple deterministic algorithm results in a better competitive ratio.

\begin{theorem}
For Pair Matching, there exists a randomized online algorithm that in expectation answers correctly on $2n/3$
input items.
\end{theorem}
\begin{proof}
Let $\alpha \in [0,1]$ be a parameter to be specified later. Intuitively, $\alpha$ denotes the probability with which our algorithm is going to answer ``reject'' on input items which are not obviously part of a pair. More specifically, suppose that the algorithm has already given answers for items $x_1, \ldots, x_{i-1}$, and a new item $x_i$ arrives. If there is a $j \in [i-1]$ such that $x_i + x_j = 1$, then the algorithm answers ``accept''. Otherwise, the algorithm answers ``reject'' with probability $\alpha$. We can analyze the performance of the algorithm by analyzing the following three groups of input items:
\begin{description}
\item[{\normalfont\textit{Input items that are not part of a pair}:}] There are $n- 2 \cdot  \pairs(x_1, \ldots, x_n)$ such input items and the algorithm answers correctly on $\alpha(n-2 \cdot \pairs(x_1, \ldots, x_n))$ in expectation.
\item[{\normalfont\textit{Input items that are the first of a pair}:}] There are $\pairs(x_1, \ldots, x_n)$ such input items and the algorithm answers correctly on $(1-\alpha)\pairs(x_1, \ldots, x_n)$ of them in expectation.
\item[{\normalfont\textit{Input items that are the last of a pair}:}] There are $\pairs(x_1, \ldots, x_n)$ such input items and the algorithm answers correctly on all of them.
\end{description}
Thus, in expectation the algorithm gives correct answers on
\[\begin{array}{cl}
  &  \alpha(n- 2 \cdot \pairs(x_1, \ldots, x_n)) + (1-\alpha) \pairs(x_1, \ldots, x_n) + \pairs(x_1, \ldots, x_n) \\[.5ex]
= & \alpha n - (3\alpha - 2) \pairs(x_1, \ldots, x_n)
\end{array}\]
items. Observe that as long as $\alpha \ge 2/3$, we can use the bound $\pairs(x_1, \ldots, x_n) \le n/2$ to derive a lower bound of $\alpha n - (3\alpha -2) n/2$ on the number of correct answers, and the largest value, $2n/3$, is attained for $\alpha=2/3$. 
Values of $\alpha$
less than $2/3$ give poorer results for the case when there are no pairs.
\end{proof}

Next, we show that the above algorithm is an optimal randomized algorithm for Pair Matching.

\begin{theorem}
For Pair Matching, no randomized online algorithm can achieve a competitive ratio less than $3/2$.
\end{theorem}
Next, we show that the above algorithm is an optimal randomized algorithm for Pair Matching.

\begin{proof}
Let $A$ be a hypothetical randomized algorithm for Pair Matching. An adversary keeps track of the current pool of possible inputs $X$. Initially, $X = \mathbb{Q} \cap [0,1]$. An adversary picks an arbitrary number $x \in X$ as the first input item. Let $p$ be the probability that $A$ answers ``reject'' on $x$. Depending on the value of $p$, there are two cases.

Case~1: $p > 2/3$, then the adversary picks $1-x$ as the next input item. One can assume that $A$ answers correctly on $1-x$. Then, the adversary removes $x$ and $1-x$ from $X$ and proceeds.

Case~2: $p \le 2/3$, then the adversary removes $x$ and $1-x$ from $X$ and proceeds.

Observe that in Case~1, the algorithm is given two input items and it answers correctly on $(1-p) + 1 = 2-p$ input items in expectation. Thus, the fraction of correct answers is $1-p/2 < 1 - 1/3 = 2/3$.

In Case~2, removing $x$ and $1-x$ from $X$ ensures that $x$ is not part of a pair in the input. Thus, the algorithm answers correctly on $p \le 2/3$ of the input in this case in expectation.
\end{proof}

Lastly, we prove that online algorithms need a lot of advice in order to start approaching a competitive ratio of~$1$ for Pair Matching.

\begin{theorem}
\label{thm:online_pm_lb}
For Pair Matching and any $\epsilon\in\EPSINT$,
no deterministic online algorithm reading fewer than $(1-H(\epsilon))n/2$ advice bits can
make fewer than $\epsilon n$ mistakes for large enough~$n$.
\end{theorem}
\begin{proof}
We prove the statement by a reduction from the 2-SGKH problem. Let $A$ be an online algorithm solving Pair Matching. Fix an arbitrary infinite sequence of distinct numbers $(y_i)_{i=1}^\infty$ from $[0,1]$.

Let $x_1, \ldots, x_n$ be the input to 2-SGKH. The online reduction works as follows. Suppose that we have already processed $x_1, \ldots, x_{i-1}$ and we have to guess the value of $x_i$. We query $A$ on $y_i$. If $A$ answers that $y_i$ is a part of a pair, then the algorithm predicts $x_i = 1$; otherwise, the algorithm predicts $x_i = 0$. Then the actual value of $x_i$ is revealed. If the actual value is $1$, then the reduction algorithm feeds $1-y_i$ as the next input item to $A$. We assume that $A$ answers correctly on $1-y_i$ in this case. If the actual value of $x_i$ is $0$, the algorithm proceeds to the next step.

Note that the number of mistakes that the reduction algorithm makes is exactly equal to the number of mistakes that $A$ makes. The statement of the theorem follows by observing that the input to $A$ is of length at most $2n$. 
\end{proof}

\subsection{Priority Setting}

In this section, we show that Theorem~\ref{thm:online_pm_lb} also holds in the priority setting. The proof becomes a bit more subtle, so we give it in full detail.

\begin{theorem}
\label{thm:pair_matching_priority_lb}
For Pair Matching and any $\epsilon\in\EPSINT$,
no fixed priority algorithm reading fewer than $(1-H(\epsilon))n/2$ advice bits can
make fewer than $\epsilon n$ mistakes for large enough~$n$.
\end{theorem}
\begin{proof}
We prove the statement by a reduction from the \emph{online problem} 2-SGKH. Let $A$ be a priority algorithm solving Pair Matching, and let $P$ be the corresponding priority function. (Note that we assume that the algorithm knows $P$; this is the case in all of our priority algorithm reductions.) The reduction follows the proof of Theorem~\ref{thm:online_pm_lb} closely. The idea is to transform the online input to 2-SGKH into an input to Pair Matching. The difficulty arises from having to present the transformed input in the online fashion while respecting the priority function $P$. 

Let $x_1, \ldots, x_n$ be the input to 2-SGKH. The online reduction works as follows. The online algorithm picks $n$ distinct numbers $y_1, \ldots, y_n$ from $[0,1]$ and creates a list $z_1, \ldots, z_{2n}$ consisting of $y_i$ and $1-y_i$ sorted according to $P$. The algorithm keeps a  (max-heap ordered) priority queue $Q$ of elements from $z_i$ as well as a subsequence $Z$ of $z_1, \ldots, z_{2n}$. The reduction always picks the first element $z$ from $Z$. We maintain the invariant that $1-z$ appears later in $Z$ according to $P$. If needed, the reduction algorithm will enter $1-z$ onto $Q$ to be simulated as an input to $A$ at the right time later on.

{\em Initialization.} Initially, $Q$ is  empty and $Z$ is the entire sequence $z_1, \ldots, z_{2n}$. Before the element $x_1$ arrives, the algorithms feeds $z_1$ to $A$. If $A$ answers that $z_1$ is a part of a pair, then the online algorithm predicts $x_1 = 1$; otherwise the algorithm predicts $x_1 = 0$. Then the online algorithm finds $j$ such that $z_j = 1-z_1$ and updates $Z$ by deleting $z_1$ and $z_j$.
Then $x_1$ is revealed. If the actual value of $x_1$ is $1$, the algorithm inserts $z_j$ into $Q$; otherwise the algorithm does not modify $Q$.

{\em Middle step.} Suppose that the algorithm has processed $x_1, \ldots, x_{i-1}$ and has to guess the value of $x_i$. The algorithm picks the first element $z$ from the subsequence $Z$. While the top element of $Q$ has higher priority than $z$ according to $P$, the algorithm deletes that element from the priority queue and feeds it to $A$. Then, the algorithm feeds $z$ to $A$. The next steps are similar to the initialization case. If $A$ answers that $z$ is a part of a pair, then the online algorithm predicts $x_i = 1$; otherwise the algorithm predicts $x_i = 0$. The online algorithm finds $z'$ in $Z$ such that $z = 1-z'$, and updates $Z$ by deleting $z$ and $z'$. Then $x_i$ is revealed. If the actual value of $x_i$ is $1$, the algorithm inserts $z'$ into $Q$; otherwise the algorithm does not modify $Q$.

{\em Post-processing.} After the algorithm finishes processing $x_n$, it feeds the remaining elements (in priority order) from $Q$ to $A$.

It is easy to see that the online algorithm feeds a subsequence of $z_1, \ldots, z_{2n}$ to $A$ in the correct order according to $P$. In addition, the online algorithm makes exactly the same number of mistakes as $A$ (assuming that $A$ always answers correctly on the second element of a pair). The statement of the theorem follows since the size of the input to $A$ is at most $2n$.
\end{proof}


\section{Reduction Template}
\label{sec:redux_temp}
Our template is restricted to binary decision problems since the goal is to derive inapproximations based on the Pair Matching problem. (See also the discussion in Section~\ref{app:bipartite_matching}.) 
In reducing from Pair Matching to a problem $B$, we assume that we 
have a priority algorithm $\ALG$  with advice, for problem $B$, with 
priorities defined by $P$. Based on $\ALG$ and $P$, we 
define a priority algorithm $\ALG'$ with advice and a priority function, $P'$,
for the  Pair Matching problem.
Input items $x_1,x_2,\ldots,x_n$ in $\mathbb{Q} \cap [0,1]$ to Pair Matching
arrive in an order specified by the priority function we define,
based on $P$.
We assume that we are informed when the input ends and can take steps
at that point to complete our computation.
Knowing the size $n$ of the input, which one naturally would in many
situations after the initial sorting according to $P'$, would of course be sufficient.

Based on the input to the Pair Matching problem,
we create input items to problem $B$, and they have to be presented
to $\ALG$, respecting the priority function $P$. Responses from $\ALG$
are then used by $\ALG'$ to help it answer ``accept'' or ``reject'' for its current
$x_i$. Actually, $\ALG$ will always answer correctly for a request 
$x_j=1-x_i$ when $i<j$, so the responses from $\ALG$ are only used when
this is not the case.
The main challenge is to ensure that the input items to $\ALG$ are presented in the order 
determined by $P$, because the decision as to whether or not they are presented
needs to be made in time, without knowing whether or not the matching value 
will arrive.

Here, we give a high level description of a specific kind of gadget 
reduction. A gadget $G$ for problem $B$ is simply some constant-sized
instance for $B$, i.e., a collection of input items that satisfy the 
consistency condition for problem $B$.
For example, if $B$ is a graph problem in the vertex arrival, vertex adjacency
model, $G$ could be a constant-sized graph,  and the universe then contains all possible pairs of the form: a vertex name coupled with a list of possible neighboring vertex names. Note that each possible vertex name exists many times as a part of an input, because it can be coupled with many different possible lists of vertex names.
The consistency condition
must apply to the actual input chosen, so for each vertex name $u$ which is listed
as a neighbor of $v$, it must be the case that $v$ is listed as a neighbor of $u$.

The gadgets used in a reduction will be created in pairs (gadgets in a pair may be
 isomorphic to each other, so that they are the same up to renaming),
one pair for each input item less than or equal to $1/2$ (for $x=1/2$,
the gadget will only be used to assign a priority to $x=1/2$).
One gadget from the pair is presented to $\ALG$ when $1-x$ appears later
in the input; and the other gadget when it does not.
Using fresh names in the input items for problem $B$, we ensure that each 
input item less than $\frac12$ to the Pair Matching problem has its own 
collection of input items for its gadgets for problem $B$.
The pair of gadgets
associated with an input item $x\leq 1/2$ can be written $(G_x^1,G_x^2)$.
The same universe of input items is used for both of these gadgets.

We write $\max_P G$ to denote the first item according to $P$ from the universe of
input items for~$G$, 
i.e., the highest priority item. 
For now, assume that $\ALG$ responds ``accept'' or ``reject'' to any possible
input item.
This captures problems such as vertex cover, independent set, clique, etc. 

For each $x\leq 1/2$, the gadget pair
satisfies two conditions: the first item condition, and 
the distinguishing decision condition. The \emph{first item condition}
says that the first input item $m_1(x)$ according to $P$ gives no information 
about which gadget it is in. To accomplish this, 
we define the priority function for $\ALG'$ as $P'(x) = P(\max_P G_x^1)$ 
for all $x\leq 1/2$ and set $m_1(x)=\max_P G_x^1 = \max_P G_x^2$ 
(the second equality holds since we assume the two gadgets
have the same input universe). The \emph{distinguishing decision condition} says 
that the decision with regards to item $m_1(x)$ that results in the 
optimal value of the objective function in $G_x^1$ is different from the 
decision that results in the optimal value of the objective function in $G_x^2$.
This explains why the one gadget is presented to $\ALG$ when 
$1-x$ appears later in the input sequence and the other when it does not.

Now that the first item of the gadget associated with $x$ is defined, 
the remaining actual input items in
the gadget pair for $x$ must be completely defined according to the 
distinguishing
decision condition. This gives two sets (overlapping, at least in $m_1(x)$)
of input items. The item with highest priority among all of the items
in the actual gadget pair, ignoring $m_1(x)$, is called $m_2(x)$, and
we define $P'(1-x)= P(m_2(x))$ for $x<1/2$. Thus, we guarantee the following
list of properties:
$x< 1/2$ will arrive before $1-x$ in the input sequence for Pair Matching for $\ALG'$,
$m_1(x)$ will arrive for algorithm $\ALG$ at the same time, $\ALG$'s
response for $m_1(x)$ can define the response of $\ALG'$ to $x$, and the
decision as to which gadget in the pair is presented for $x$ can be made
at the time $1-x$ arrives or $\ALG'$ can determine that it will not arrive
(because either the input sequence ended or an $x'$ with lower priority
than $1-x$ arrived).

To warm up, we start with an example reduction from Pair Matching to a somewhat artificial problem. This reduction then serves as a model for the general reduction template.

\subsection{Example: Triangle Finding}

Consider the following priority problem in the vertex arrival, vertex adjacency 
model: for each vertex $v$, decide whether or not $v$ belongs to some triangle (a cycle 
of length $3$) in the entire input graph. The answer
``accept'' is correct if $v$ belongs to some triangle, and otherwise the answer
should be ``reject''.
We refer to this problem as
Triangle Finding. This problem might look artificial and it is optimally solvable  
offline in time $O(n^2)$, but as mentioned above, advice-preserving 
reductions between priority problems require subtle manipulations of a 
priority function.
The Triangle Finding problem allows us to highlight this issue in a relatively 
simple setting.

\begin{theorem}
For Triangle Finding and any $\epsilon\in\EPSINT$,
no fixed priority algorithm reading at most $(1-H(\epsilon))n/8$ advice bits can
make fewer than $\epsilon n/4$ mistakes.
\end{theorem}
\begin{proof}
We prove this theorem by a reduction from the Pair Matching problem. Let $\ALG$
be an algorithm for the Triangle Finding problem, and let $P$ be the 
corresponding priority function.
Let $x_1, \ldots, x_n$ be the input to 
Pair Matching. We define a priority function $P'$ and a valid input 
sequence $v_1, \ldots, v_m$ to Triangle Finding. When $x_1, \ldots, x_n$ is 
presented according to $P'$ to our priority algorithm for Pair Matching, it 
is able to construct $v_1, \ldots, v_m$ for $\ALG$, respecting the priority 
function $P$. Moreover, our algorithm for Pair Matching will be able to use 
answers of $\ALG$ to answer the queries about $x_1, \ldots, x_n$. 

Now, we discuss how to define $P'$. With each number $x \in \mathbb{Q} \cap [0,1/2]$,
we associate four unique vertices $v_x^1, v_x^2, v_x^3, v_x^4$.
The universe consists of all input items of the form 
$(v_x^i, \{v_x^j, v_x^k\})$ with $i, j, k \in [4]$, $i \not\in \{ j,k\}$ and $j < k$;
there are $12$ input items for each $x$: $4$ possibilities for the vertex, and for each of the 
$\binom{3}{2}=3$ possibilities for the ordered pair of neighbors.
Let $m_1(x)$ be the first item according to $P$ among the $12$ items.
Using only the input items from the $12$ items we are currently considering,
we extend this item in two ways, to a 3-cycle $C_x^3$ and to a 4-cycle $C_x^4$.
When we write $C_x^3$ or $C_x^4$, we mean the set of items forming the 3-cycle or 4-cycle,
respectively. Now, $P'$ is defined as follows:
\[
P'(x) = \left\{\begin{array}{ll}
             P(m_1(x)), & \mbox{if $x \le 1/2$} \\
             \max_{ g \in (C_{1-x}^3 \cup C_{1-x}^4) \setminus \{ m_1(1-x)\}} P(g), & \mbox{otherwise}
               \end{array}
        \right.
\]
In other words, if $x > 1/2$, we set $P'(x)$ to be the first element other 
than $m_1(1-x)$ in $C_{1-x}^3 \cup C_{1-x}^4$. 
In terms of our high level description given at the beginning of this section, 
$(C_x^3, C_x^4)$ form the pair of gadgets -- a triangle and a square. By 
construction, this pair of gadgets satisfies the first item condition. By the 
definition of the problem, the optimal decision for all vertices in $C_x^3$ is 
``accept'' (belongs to a triangle) and the optimal decision for all vertices in 
$C_x^4$ is ``reject'' (does not belong to a triangle). Thus, these gadgets also 
satisfy the distinguishing decision condition.
 
Let $x_1, \ldots, x_n$ denote the order input items are presented to our algorithm
as specified by $P'$.
Our algorithm constructs an input to $\ALG$ which is consistent 
with $P$ along the following lines: for each $x \le 1/2$ that appears 
in the input, the algorithm constructs either a three-cycle or a four-cycle 
(disjoint from the rest of the graph). Thus, each $x \le 1/2$ is associated 
with one
connected component. During the course of the algorithm, each connected 
component will be in one of the following three states: undecided, committed, 
or finished. When $x \le 1/2$ arrives, the algorithm initializes the 
construction with the item $m_1(x)$ and sets the component status to undecided. 
It answers ``accept'' (there will be a matching pair) for $x$ if $\ALG$ 
responds ``accept'' (triangle) for $m_1(x)$, and it answers ``reject'' if $\ALG$ 
responds ``reject'' (square).

Note that for any $x\leq 1/2$, $P'(x)>P'(1-x)$, so if $x' > 1/2$ arrives
and $1-x'$ has not appeared earlier, $\ALG'$ can simply reject $x'$ and
does not need to present anything to $\ALG$. If $x$ has arrived and at some point, $1-x$ arrives, the algorithm commits 
to constructing the 
3-cycle $C_x^3$. If $\ALG'$ had guessed correctly that $1-x$ would arrive,
it is because $\ALG$ responded ``accept'' for $m_1(x)$) and also guessed correctly.
If $\ALG'$ had guessed that $1-x$ would not arrive, it is because $\ALG$ guessed that
a square would arrive, and both guessed incorrectly. If some $x'$ arrives with $P'(x')<P'(1-x)$ for some $x\not=x'$ and $x$ has arrived earlier,
then $\ALG'$ can be certain that $1-x$ will not arrive.
It commits to constructing the 4-cycle $C_x^4$.
Thus, if $\ALG'$ answered ``reject'' for $x$, it answered correctly,
and a square makes $\ALG$'s decision
for $m_1(x)$ correct. Similarly, if $\ALG'$ answered ``accept'' for $x$,
it answered incorrectly, so a square makes $\ALG$'s decision incorrect.

At the end of the input, $\ALG'$ finishes off by
checking which values of $x$ have arrived without $1-x$ arriving or
some $x'$ with higher priority than $1-x$ arriving, and $\ALG$
again commits to the 4-cycle, as in the other case where $1-x$ does not
arrive.

Throughout the algorithm, there are several connected components, each of which can
be undecided, committed, or finished.
Note that an undecided component corresponding to input $x$ consists of a single item $m_1(x)$. 
Upon receiving an item $y$, the algorithm first checks whether some undecided components have turned  into committed ones: namely if an undecided component consisting of $m_1(x)$ satisfies $P'(1-x) > P'(y)$, it switches the status to a committed component according to the rules described above. Then, the algorithm feeds input items corresponding to committed yet unfinished connected components to $\ALG$ and does so in the order of $P$ up until the priority of such items falls below $P'(y)$ (this can be done by maintaining a priority queue). Finally, the algorithm processes the item $y$ by either creating a new component or by turning an undecided component into a decided one. Then, the algorithm moves to the next item. Due to our definition of $P'$ and this entire process, the input constructed for $\ALG$ is valid and consistent with $P$. 
Observe that the input to $\ALG'$ is of size at most $4n$, so the number of advice bits
must be divided by four relative to Theorem~\ref{thm:pair_matching_priority_lb},
and the theorem follows.
\end{proof}

\subsection{General Template}
\label{sec:general_template}

In this subsection, we establish two theorems that give general templates for gadget reductions from Pair Matching -- one for maximization problems and one for minimization problems. The high level overview has been given at the beginning of this section.

We let $\ALG(I)$ denote the objective function for $\ALG$ on input $I$. The \emph{size} of a gadget $G$, denoted by $\SIZE{G}$, is the number of input items specifying the gadget. We write $\OPT(G)$ to denote the best value of the objective function on $G$. Recall that we focus on problems where a solution is specified by making an accept/reject decision for each input item. We write $\BAD(G)$ to denote the best value of the objective function attainable on $G$ after making the wrong decision for the first item (the item with highest priority, $\max(G)$), i.e., if there is an optimal solution that accepts (rejects) the first item of $G$, then $\BAD(G)$ denotes the best value of the objective function \emph{given} that the first item was rejected (accepted). We say that the objective function for a problem $B$ is \emph{additive}, if for any two instances $I_1$ and $I_2$ to $B$ such that $I_1 \cap I_2 = \emptyset$, we have $\OPT(I_1 \cup I_2) = \OPT(I_1) + \OPT(I_2)$.

\begin{theorem}
\label{thm:template-minimization}
Let $B$ be a minimization problem with an additive objective function. Let $\ALG$ be a fixed priority algorithm with advice for $B$ with a priority function $P$. Suppose that for each $x \in \mathbb{Q} \cap [0, 1/2]$ one can construct a pair of gadgets $(G_x^1, G_x^2)$ satisfying the following conditions:
\begin{description}
\item[{\normalfont\textit{The first item condition}:}] $m_1(x)=\max_P G_x^1 = \max_P G_x^2$.
\item[{\normalfont\textit{The distinguishing decision condition}:}] the optimal decision for $m_1(x)$ in $G_x^1$ is different from the optimal decision for $m_1(x)$ in $G_x^2$ (in particular, the optimal decision is unique for each gadget). Without loss of generality, we assume $m_1(x)$ is accepted in an optimal solution in $G_x^1$.
\item[{\normalfont\textit{The size condition}:}] the gadgets have finite sizes, and we let $s=\max_x (|G_x^1|, |G_x^2|)$, where the cardinality of a gadget is the number of input items it consists of.
\item[{\normalfont\textit{The disjoint copies condition}:}] for $x\not=y$ and $i,j\in\SET{1,2}$, input items making up $G_x^i$ and $G_y^j$ are disjoint.
\item[{\normalfont\textit{The gadget $\OPT$ and $\BAD$ condition}:}] the values $\OPT(G_x^1)$, $\BAD(G_x^1)$ as well as $\OPT(G_x^2)$, $\BAD(G_x^2)$ are independent of $x$, and we denote them by $\OPT(G^1)$, $\BAD(G^1)$, $\OPT(G^2)$, and $\BAD(G^2)$; we assume that $\OPT(G^2) \ge \OPT(G^1)$.
\end{description}
Define $r=\min\SET{\frac{\BAD(G^1)}{\OPT(G^1)}, \frac{\BAD(G^2)}{\OPT(G^2)}}$. Then for any $\epsilon \in (0,\frac12)$, no fixed priority algorithm reading fewer than $(1-H(\epsilon))n/(2s)$ advice bits can achieve an approximation ratio smaller than 
\[1 + \frac{\epsilon (r-1) \OPT(G^1)}{\epsilon \OPT(G^1) + (1-\epsilon)\OPT(G^2)}.\]
\end{theorem}
\begin{proof}
The proof proceeds by constructing a reduction algorithm (fixed priority with advice) for Pair Matching that uses $\ALG$ to make decisions about input items. We start by defining a priority function for the reduction algorithm.

Define $m_2(x)$ to be the highest priority input item in $G_x^1$ or $G_x^2$ different from $m_1(x)$, i.e., 
\[m_2(x) = \max \left(\left(G_x^1 \cup G_x^2\right) \setminus \SET{m_1(x)}\right).\]

We define a priority function $P'$ as follows.

\[P'(x) = 
 \left\{\begin{array}{ll}
    P(m_1(x)), & \mbox{if $x\le \frac12$} \\[1ex]
    P(m_2(1-x)), & \mbox{if $x>\frac12$}
 \end{array}\right.\]

\noindent For the Pair Matching problem, we denote the given input sequence ordered by $P'$ as $I=\SEQ{x_1, \ldots, x_n}$. We have to give an overall strategy for how the reduction algorithm for Pair Matching handles an input item $x_i$ and which input items it presents to $\ALG$. In order to do this, we use a priority queue $Q$ which is a max-heap ordered based on the priority of input items to problem~$B$, with the purpose of presenting these input items in the correct order (respecting $P$, highest priority items appear first). When $\ALG'$ commits to a particular gadget in a pair, the remainder of that gadget (all inputs except $m_1(x)$ which has already been presented) are inserted into $Q$.

\begin{algorithm}
\caption{Reduction Algorithm, $\ALG'$}\label{alg:template}
\textbf{Given:} $\ALG$ with priority function $P$ for problem~$B$ \\
\begin{algorithmic}[1]
\State $Q.\init()$ \Comment{Initialize $Q$ to empty}
\For {$i = 1, \ldots, n$}
  \If {$x_i \geq \frac12$}
    \If {$x_i=1-x_j$ for some $j < i$}
      \State accept $x_i$
      \State insert $G_{x_j}^1 \setminus \SET{m_1(x_j)}$ into $Q$ \label{line-match-one}
    \Else
      \State reject $x_i$
    \EndIf
  \EndIf
  \For {all $1\leq j<i$ s.t.\ $P'(x_{i-1}) > P'(1-x_j) > P'(x_i)$} \Comment{no $1-x_j$} \label{line-test}
    \State insert $G_{x_j}^2 \setminus \SET{m_1(x_j)}$ into $Q$ \label{line-non-match}
  \EndFor
  \While {$Q.\findmax() > P'(x_i)$}
     \State present $Q.\deletemax()$ to $\ALG$ \label{line-catch-up}
  \EndWhile
  \If {$x_i< \frac12$}
     \State present $m_1(x_i)$ to $\ALG$ \label{line-present-mone}
     \State answer the same as $\ALG$ \label{line-answer-same}
  \EndIf
\EndFor
\For {all $1\leq j\leq n$ s.t.\ $P'(1-x_j) \leq  P'(x_n)$} \Comment{no $1-x_j$} \label{line-test-two}
  \State insert $G_{x_j}^2 \setminus \SET{m_1(x_j)}$ into $Q$ \label{line-non-match-two}
\EndFor
\While {not $Q.\isempty()$}
   \State present $Q.\deletemax()$ to $\ALG$
\EndWhile
\end{algorithmic}
\end{algorithm}

By definition, $P'(x_i) > P'(1-x_i)$ for all $x_i < 1/2$. Thus, $m_1(x_i)$ is presented to $\ALG$ in Line~\ref{line-present-mone} before the remaining parts of the same gadget \emph{associated with} $x_i$ are inserted into $Q$ in one of Lines~\ref{line-match-one}, \ref{line-non-match},
or~\ref{line-non-match-two}.

Since the priority of any $x_i < \frac12$ is defined to be the priority of $m_1(x_i)$, the $m_1(x_i)$s are presented in the correct relative order.

Clearly, input items entered into the priority queue, $Q$, are extracted and presented to $\ALG$ in the correct relative order, and before any $m_1(x_i)$ is presented, higher priority items are presented first in Line~\ref{line-catch-up}. The remaining issues are whether the remainder of the gadget
associated with some $x_j$ is entered into $Q$ early enough relative to some $m_1(x_i)$ from another gadget and whether all gadgets are eventually completely presented to $\ALG$.

By the definition of $m_2$, the priority of $m_2(x_j)$ is at least the priority of any remaining input item in the gadget associated with $x_j$.

Consider the point in time when $x_i$ arrives.
If $1-x_j$ arrived earlier or $P'(1-x_j)$ is greater than $P'(x_{i-1})$, the gadget associated with $x_j$ would have been processed correctly or have been inserted into $Q$ earlier. Before $m_1(x_i)$ is presented to $\ALG$, a check is made to see if $P(m_2(x_j))=P'(1-x_j)>P'(x_i)=P(m_1(x_i))$. If the check in the \textbf{if}-statement is positive,
the entire remaining part of gadget for $x_j$ is inserted into $Q$ at this point
in Line~\ref{line-non-match}.

If some $x_j<\frac12$ arrives, but $1-x_j$ never arrives, if $P'(1-x_j)\leq 
P'(x_n)$, this is discovered in Line~\ref{line-test-two} and the remainder of $G^2_{x_j}$ is
presented to $\ALG$ in Line~\ref{line-non-match-two}.

Thus, input items are presented to $\ALG$ in the order defined by its priority function~$P$.

Now we turn to the approximation ratio obtained. We want to lower bound the number of incorrect decisions by $\ALG$. We focus on the input items which are $m_1(x_i)$ for some input $x_i < 1/2$ to the Pair Matching Problem and assume that $\ALG$ answers correctly on anything else.

When $\ALG'$ receives an $x_i < 1/2$, in Line~\ref{line-answer-same} it answers the same for $x_i$ as $\ALG$ does for $m_1(x_i)$. By considering
the four cases where the gadget associated with $x_i$ is later inserted into $Q$, we can see that this answer for $x_i$ was correct for $\ALG'$
if and only if the answer $\ALG$ gave for $m_1(x_i)$ could lead to the optimal result for the gadget associated with $x_i$.

\begin{itemize}

\item If $x_i=1-x_j$ arrives, 
then $G^1_{x_j}$ is committed to and the remainder of $G^1_{x_j}$ is inserted into $Q$ in Line~\ref{line-match-one}. 
If $\ALG'$ answered ``accept'' to $x_j$, then $\ALG$ has accepted $m_1(x_j)$
and  $\ALG$ could obtain the optimal result on  $G^1_{x_j}$, by the definition of these gadget pairs.
If $\ALG'$ answered ``reject'' to $x_j$, then $\ALG$ has rejected $m_1(x_j)$
and  $\ALG$ cannot obtain the optimal result on  $G^1_{x_j}$, again by the definition of these gadget pairs.

\item If $x_i=1-x_j$ does not arrive, 
then $G^2_{x_j}$ is committed to and the remainder of $G^2_{x_j}$ is inserted into $Q$ in Lines~\ref{line-non-match} or \ref{line-non-match-two}. 
If $\ALG'$ answered ``reject'' to $x_j$, then $\ALG$ has rejected $m_1(x_j)$
and  $\ALG$ could obtain the optimal result on  $G^2_{x_j}$, by the definition of these gadget pairs.
If $\ALG'$ answered ``accept'' to $x_j$, then $\ALG$ has accepted $m_1(x_j)$
and  $\ALG$ cannot obtain the optimal result on  $G^2_{x_j}$, again by the definition of these gadget pairs.

\end{itemize}

 We know from Theorem~\ref{thm:pair_matching_priority_lb} that for any $\epsilon \in (0, 1/2]$, any priority algorithm with advice length  less than $(1-H(\epsilon))n/2$ makes at least $\epsilon n$ mistakes. Since we want to lower bound the performance ratio of $\ALG$, and since a ratio larger than one decreases when increasing the numerator and denominator by equal quantities, we can assume that when $\ALG$ answers correctly, it is on the gadget with the larger $\OPT$ -value, $G^2$. For the same reason, we can assume that the ``at least $\epsilon n$'' incorrect answers are in fact exactly $\epsilon n$, since classifying some of the incorrect answers as correct just lowers the 
ratio. For the incorrect answers, assume that the gadget $G^1$ is presented $w$ times, and, thus, the gadget, $G^2$, $\epsilon n - w$ times.

Denoting the input created by $\ALG'$ for $\ALG$ by $I$, we obtain the following, where we use that $\BAD(G^j) \geq r\OPT(G^j)$.
\[\begin{array}{rcl}

\frac{\ALG(I)}{\OPT(I)} & \geq &

\frac{(1-\epsilon)n\OPT(G^2)+w\BAD(G^1)+(\epsilon n - w)\BAD(G^2)}{
  (1-\epsilon)n\OPT(G^2)+w\OPT(G^1)+(\epsilon n - w)\OPT(G^2)}
\\[1.5ex]
& \geq &

\frac{(1-\epsilon)n\OPT(G^2)+wr\OPT(G^1)+(\epsilon n - w)r\OPT(G^2)}{
  (1-\epsilon)n\OPT(G^2)+w\OPT(G^1)+(\epsilon n - w)\OPT(G^2)}
\\[1.5ex]
& = &

1 + \frac{w(r-1)\OPT(G^1) + (\epsilon n -w)(r-1)\OPT(G^2)}{
  w\OPT(G^1) + (n-w)\OPT(G^2)}

\end{array}\]

Taking the derivative with respect to $w$ and setting equal to zero gives no solutions for $w$, so the extreme values must be found at the endpoints of the range for $w$ which is $[0,\epsilon n]$.

Inserting $w=0$, we get $1+\epsilon(r-1)$, while $w=\epsilon n$ gives
\[1 + \frac{\epsilon (r-1) \OPT(G^1)}{\epsilon \OPT(G^1) + (1-\epsilon)\OPT(G^2)}.\]

The latter is the smaller ratio and thus the lower bound we can provide.

\end{proof}

The following theorem for maximization problems is proved analogously.

\begin{theorem}
\label{thm:template-maximization}
Let $B$ be a maximization problem with an additive objective function. Let $\ALG$ be a fixed priority algorithm with advice for $B$ with a priority function $P$. Suppose that for each $x \in \mathbb{Q} \cap [0, 1/2]$ one can construct a pair of gadgets $(G_x^1, G_x^2)$ satisfying the conditions in Theorem~\ref{thm:template-minimization}.
Set $r=\min\SET{\frac{\OPT(G^1)}{\BAD(G^1)}, \frac{\OPT(G^2)}{\BAD(G^2)}}$. Then for any $\epsilon \in (0,\frac12)$, no fixed priority algorithm reading fewer than $(1-H(\epsilon))n/(2s)$ advice bits can achieve an approximation ratio smaller than \[1 + \frac{\epsilon (r-1) \OPT(G^1)}{\epsilon \OPT(G^1) + (1-\epsilon)r\OPT(G^2)}.\]
\end{theorem}

\begin{proof}
The proof proceeds as for the minimization case in Theorem~\ref{thm:template-minimization}
until the calculation of the lower bound of $\frac{\ALG(I)}{\OPT(I)}$.
We continue from that point, using the inverse ratio to get values larger than one.

We use that $\BAD(G^j) \leq \OPT(G^j) / r$.
\[\begin{array}{rcl}

\frac{\OPT(I)}{\ALG(I)} & \geq &

\frac{(1-\epsilon)n\OPT(G^2)+w\OPT(G^1)+(\epsilon n - w)\OPT(G^2)}{
      (1-\epsilon)n\OPT(G^2)+w\BAD(G^1)+(\epsilon n - w)\BAD(G^2)}
\\[1.5ex]
& \geq &

\frac{(1-\epsilon)n\OPT(G^2)+w\OPT(G^1)+(\epsilon n - w)\OPT(G^2)}{
      (1-\epsilon)n\OPT(G^2)+\frac{w}{r}\OPT(G^1)+\frac{\epsilon n - w}{r}\OPT(G^2)}

\end{array}\]

Again, taking the derivative with respect to $w$ gives an
always non-positive result. Thus, the smallest value in
the range $[0,\epsilon n]$ for $w$ is found at $w=\epsilon n$.
Inserting this value, we continue the calculations from above:

\[\begin{array}{rcl}

\frac{\OPT(I)}{\ALG(I)} & \geq &
\frac{(1-\epsilon)n\OPT(G^2)+w\OPT(G^1)+(\epsilon n - w)\OPT(G^2)}{
      (1-\epsilon)n\OPT(G^2)+\frac{w}{r}\OPT(G^1)+\frac{\epsilon n - w}{r}\OPT(G^2)}
\\[1.5ex]

& = &
\frac{(1-\epsilon)n\OPT(G^2)+(\epsilon n)\OPT(G^1)}{
      (1-\epsilon)n\OPT(G^2)+\frac{\epsilon n}{r}\OPT(G^1)}
\\[1.5ex]

& = &
\frac{(1-\epsilon)r\OPT(G^2)+\epsilon r\OPT(G^1)}{
      (1-\epsilon)r\OPT(G^2)+\epsilon\OPT(G^1)}
\\[1.5ex]

& = &
1 + \frac{\epsilon (r-1) \OPT(G^1)}{
          (1-\epsilon)r\OPT(G^2)+\epsilon\OPT(G^1)}

\end{array}\]

The latter is the smaller ratio and thus the lower bound we can provide.

\end{proof}

We mostly use
Theorems~\ref{thm:template-minimization} and~\ref{thm:template-maximization}
in the following specialized form.

\begin{corollary}
\label{corollary-form-used}
With the set-up from 
Theorems~\ref{thm:template-minimization} and~\ref{thm:template-maximization},
we have the following:

For a minimization problem, if $\OPT(G^1)=\OPT(G^2)=\BAD(G^1)-1=\BAD(G^2)-1$,
then no fixed priority algorithm reading fewer than $(1-H(\epsilon))n/(2s)$
advice bits can achieve an approximation ratio smaller than
$1+\frac{\epsilon}{\OPT(G^1)}$.

For a maximization problem, if $\OPT(G^1)=\OPT(G^2)=\BAD(G^1)+1=\BAD(G^2)+1$,
then no fixed priority algorithm reading fewer than $(1-H(\epsilon))n/(2s)$
advice bits can achieve an approximation ratio smaller than
$1+\frac{\epsilon}{\OPT(G^1)-\epsilon}$.
\end{corollary}

Next, we describe a general procedure for constructing gadgets with the above properties. For simplicity, we do it for graph problems in the  vertex arrival, vertex adjacency input model. Later we discuss what is required to carry out such general constructions for other combinatorial problems. In the case of graphs, an input item consists of a vertex name with the names of neighbors of that vertex. First, consider defining a single gadget instead of a pair. We define a gadget in several steps. As the first step, we define a graph $G=\left([n], E \subset {[n] \choose 2}\right)$ over $n$ vertices.
Then, when defining a gadget based on input $x$ to Pair Matching, we pick $n$ vertex names $V_x$ and give a bijection $f\colon V_x \rightarrow [n]$.
Finally, we read off the resulting input items in the order given by the priority function. Thus, we think of $G$ as giving a topological structure of the instance, and it is converted into an actual instance by assigning new names to the vertices. The reason that the names from the topological structure are not used directly is that we want to define a separate gadget instance for each $x \in \mathbb{Q} \cap [0,1/2]$. Thus, all gadgets instances are going to have the same topological structure\footnote{However, both gadgets within a pair do not necessarily have the same topological structure. In Triangle Finding, they did not.}, but will differ in names of vertices. 

For graphs in the vertex arrival, vertex adjacency model, we say that two input items are isomorphic if they have the same number of neighbors, i.e., they differ in just the names of the vertices and the names of their neighbors. A topological structure $G$ consisting only of isomorphic items is a regular graph. For any priority function $P$ and any vertex $v \in [n]$, we can force the corresponding item to appear first according to $P$ by naming vertices appropriately. Fix $x$ and consider all possible input items that can be formed from $V_x$ consistently with $G$.
One of those items appears first according to $P$. Define a bijection $f$ by first mapping that first item to $u$ and its neighbours in $G$, and extending this one-to-one correspondence to other vertices in $G$ in an arbitrary, consistent manner. In this case, the input item corresponding to $u$ would appear first according to $P$ in the input to the graph problem. Because all items are isomorphic, it is always possible to extend the bijection to all of $G$.

Now, suppose that two topological structures $G^1=([n], E^1)$ and $G^2=([m],E^2)$ consist only of isomorphic items. Using a similar idea, for each priority function $P$, each $x \in \mathbb{Q} \cap [0,1/2)$, each $u \in [n]$, and each $v \in [m]$, one can assign names to vertices of $G^1$ and $G^2$ such that the first input item according to $P$ is associated with $u$ in $G^1$ and the same item is associated with $v$ in $G^2$. In particular, this means that as long as the two topological structures are regular, we can always convert them into gadgets satisfying the first item condition.

Suppose that there is a vertex $u$ in $G^1$ that appears in every optimal solution in $G^1$, i.e., a ``reject'' decision leads to non-optimality. Furthermore, suppose that there is a vertex $v$ in $G^2$ that is excluded from every optimal solution in $G^2$, i.e., an ``accept'' decision leads to non-optimality. Then for each $x$, using the above construction, we can make the first item according to $P$ be associated with $u$ in $G^1$ and with $v$ in $G^2$. This means that we can always convert the topological structures into gadgets satisfying the distinguishing decision condition. Finally, observe that the size condition is satisfied with $s = \max(|G^1|, |G^2|)$.

We note a very important special case of the above
construction. Suppose that a single topological structure $G$ that
consists solely of isomorphic input items is such that the optimal
solution is unique and non-trivial, i.e., both ``accept'' and
``reject'' decisions must be represented in the optimal solution. Then
we can duplicate $G$ and pick $u$ to be a vertex which is accepted in
the unique solution and $v$ to be a vertex which is rejected in the
unique solution, and apply the above construction. All in all, this
reduces the problem of defining gadgets to finding a small regular graph
with a unique, non-trivial optimal solution. The size of such a graph
is then equal to the parameter $s$ in
Theorems~\ref{thm:template-minimization}
and~\ref{thm:template-maximization} and
Corollary~\ref{corollary-form-used}. One can relax the condition of a
unique solution and require that the topological gadget has an input
item $u$ with decision ``accept'' in every optimal solution, and an
input item $v$ with decision ``reject'' in every optimal
solution.

This gadget construction can clearly be carried out in other input models. There are very few requirements: we need to have a notion of isomorphism between input items, and a notion of the topological structure of a gadget. Once we have those two notions, if we find a topological structure consisting only of isomorphic items with a unique, non-trivial optimal solution, then we immediately conclude that the problem requires the tradeoff between advice and approximation ratio as outlined in Theorems~\ref{thm:template-minimization} and~\ref{thm:template-maximization} and Corollary~\ref{corollary-form-used} with parameter $s$ equal to the size of the topological template.

We finish this section by remarking that one can perform similar reductions with gadgets
where not all input items are isomorphic. In our last example in Section~\ref{VC},
we present a reduction for Vertex Cover using \emph{two} gadget pairs with some vertices of degree~$2$
and others of degree~$3$. One simply needs that there is one gadget pair for the case where
a vertex of degree~$2$ has the highest priority and another gadget pair for the case where a
vertex of degree~$3$ has highest priority. For both gadget pairs, $s=7$, the optimal value
is $3$, and the minimum possible objective value for the gadget in the pair is $4$. Thus,
the results of Theorem~\ref{thm:template-minimization} (or Theorem~\ref{thm:template-maximization}
if it was a mazimization problem) and Corollary~\ref{corollary-form-used} can be applied.
This idea can be extended to other input models where the gadgets have input items which are
not isomorphic. For simplicity, we do not restate the two theorems or the corollary for
the extension where there are $t$ different classes of isomorphic input items and thus $t$
pairs of gadgets.


\section{Reductions to Classic Optimization Problems}
\label{sec:classical_redux}
In this section, we provide  examples of applications of the general reduction template. 
With the exception of bipartite matching, all of these problems are NP-hard, as a consequence of the NP-completenes of their underlying decision problems, as established in the seminal papers by Cook~\cite{Cook71} and Karp~\cite{Karp72}. Furthermore, these problems are known to have various hardness of approximation bounds. 

\subsection{Independent Set}
\label{sec:is}

First, we consider the maximum independent set problem in the vertex arrival, vertex adjacency input model.
Consider the topological structure of a gadget in Figure~\ref{fig:is-gadget}. There are 5 vertices on the top and 3 vertices on the bottom. All top vertices are connected to all bottom vertices. Additionally, the 5 vertices on the top form a cycle. In this way, each vertex has degree 5 and hence all the input items are isomorphic. If we pick any vertex from the top to be in the independent set, then we forgo all the bottom vertices, and we are essentially restricted to picking an independent set from $C_5$, which has size at most 2. On the other hand, we could pick all 3 vertices from the bottom to form an independent set. 

Suppose without loss of generality that the highest priority input item
is $(1, \{ 4,5,6,7,8\} )$.
The optimal decision for the first vertex is unique: For $G^1$, one should accept,
and for $G^2$, reject.

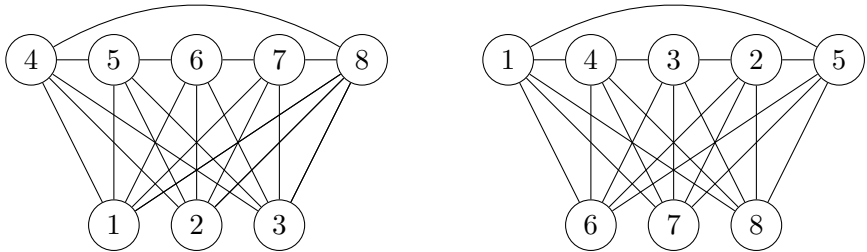
\begin{figure}[h]
\centering

\begin{tikzpicture}[scale=0.55, every node/.style={scale=1.0}]

\node[draw=black,circle,fill=white] (4) at (0,0) {4};
\node[draw=black,circle,fill=white] (5) at (2,0) {5};
\node[draw=black,circle,fill=white] (6) at (4,0) {6};
\node[draw=black,circle,fill=white] (7) at (6,0) {7};
\node[draw=black,circle,fill=white] (8) at (8,0) {8};

\node[draw=black,circle,fill=white] (1) at (2,-4) {1};
\node[draw=black,circle,fill=white] (2) at (4,-4) {2};
\node[draw=black,circle,fill=white] (3) at (6,-4) {3};

\draw (6) -- (1);
\draw (6) -- (2);
\draw (6) -- (3);

\draw (8) -- (1);
\draw (8) -- (2);
\draw (8) -- (3);

\draw (4) -- (1);
\draw (4) -- (2);
\draw (4) -- (3);

\draw (5) -- (1);
\draw (5) -- (2);
\draw (5) -- (3);

\draw (7) -- (1);
\draw (7) -- (2);
\draw (7) -- (3);

\draw (8) -- (1);
\draw (8) -- (2);
\draw (8) -- (3);

\draw (6) -- (5);
\draw (5) -- (4);
\draw (6) -- (7);
\draw (7) -- (8);
\path (8) edge[bend right] (4);

\end{tikzpicture} 
\hspace{1cm}
\begin{tikzpicture}[scale=0.55, every node/.style={scale=1.0}]

\node[draw=black,circle,fill=white] (1) at (0,0) {1};
\node[draw=black,circle,fill=white] (4) at (2,0) {4};
\node[draw=black,circle,fill=white] (3) at (4,0) {3};
\node[draw=black,circle,fill=white] (2) at (6,0) {2};
\node[draw=black,circle,fill=white] (5) at (8,0) {5};

\node[draw=black,circle,fill=white] (6) at (2,-4) {6};
\node[draw=black,circle,fill=white] (7) at (4,-4) {7};
\node[draw=black,circle,fill=white] (8) at (6,-4) {8};

\draw (1) -- (6);
\draw (1) -- (7);
\draw (1) -- (8);

\draw (2) -- (6);
\draw (2) -- (7);
\draw (2) -- (8);

\draw (3) -- (6);
\draw (3) -- (7);
\draw (3) -- (8);

\draw (4) -- (6);
\draw (4) -- (7);
\draw (4) -- (8);

\draw (5) -- (6);
\draw (5) -- (7);
\draw (5) -- (8);

\draw (1) -- (4);
\draw (2) -- (3);
\draw (3) -- (4);
\draw (2) -- (5);
\path (5) edge[bend right] (1);

\end{tikzpicture} 

\caption{Topological structure of the gadgets $(G^1,G^2)$ for independent set.}\label{fig:is-gadget}
\end{figure}

In this case, the maximum number $s$ of input items for a gadget is $8$, $\OPT(G^1)=\OPT(G^2)=3$, and $\BAD(G^1)=\BAD(G^2)=2$. By Corollary~\ref{corollary-form-used}, we can conclude the following:

\begin{theorem}
\label{thm:IS_lb}
For Maximum Independent Set and any $\epsilon\in\EPSINT$,
no fixed priority algorithm reading fewer than $(1-H(\epsilon))n/16$ advice bits can
achieve an approximation ratio smaller than $1+\frac{\epsilon}{3-\epsilon}$.
\end{theorem}
Theorem~\ref{thm:IS_lb} is related to but incomparable with the inapproximation bound results on priority algorithms (without advice) of Borodin et al.~\cite{BorodinBLM2010} for weaker models.

\subsection{Bipartite Matching}
\label{app:bipartite_matching}

Given a bipartite graph $G = (U, V, E)$ where $E \subseteq U \times V$, a matching in $G$ is a collection of vertex disjoint edges. For maximum bipartite matching, we must find a matching of maximum cardinality. In this section, we consider the maximum bipartite matching problem in vertex arrival, vertex adjacency model. In this model, an input item consists of a vertex name \emph{necessarily} from $U$ together with names of neighbors \emph{necessarily} in $V$. Thus, the $U$-side can be considered to be ``online'' and the whole graph $G$ is revealed one vertex from $U$ at a time.

Note that our framework was stated to work for decisions over a binary alphabet $\Sigma=\{\text{``accept''},$ $\text{``reject''}\}$. Strictly speaking, in bipartite matching, decisions are stated most naturally over a larger alphabet. For instance, consider an input item $(u, \{v_1, \ldots, v_k\})$, then the decision can be thought of as being made over an alphabet $\Gamma=V \cup \{\bot\}$. Here, a decision $v$ stands for matching $u$ with $v$, and a decision $\bot$ stands for not matching $u$ at all. We can still apply our framework to bipartite matching by surjectively mapping $\Gamma$ onto $\Sigma$ via $f$ as follows: $f(v) = \text{``accept''}$, $f(\bot) = \text{``reject''}$. In effect, we convert a priority algorithm with decisions over $\Gamma$ into a priority algorithm with decisions over $\Sigma$. Since we are interested in lower bounds, the result for $\Sigma$ carries over to $\Gamma$. Of course, this idea is not specific to bipartite matching, and similar alphabet transformations can be done for all problems with decisions over non-binary alphabets. It is reasonable to believe that a framework applicable directly to non-binary alphabets could be used to derive stronger inapproximation results.

Following the reduction template, two input items are isomorphic if the corresponding vertices have the same degree. Thus, a gadget consists of isomorphic items if it is a bipartite graph that is regular on the $U$-side, whereas there are no requirement for the $V$-side. Consider the topological structure of the $3$ by $3$ gadgets in Figure~\ref{fig:bm-gadget}, where
$G^1 = ([3],[3],E^1)$ with $E^1 = \{(1,1),(1,2),(2,2),(2,3),(3,2),(3,3)\}$ and
$G^2 = ([3],[3],E^1)$ with $E^2 = \{(1,1),(1,2),(2,1),(2,3),(3,1),(3,3)\}$. 
All input items are isomorphic -- they are vertices of degree $2$. 
Suppose without loss of generality that the highest priority input item is $(1, \{ 1,2\})$.
The optimal decision for the first vertex is unique: For $G^1$ choose the edge
$(1,1)$, and for $G^2$ choose $(1,2)$.

\begin{figure}[h]
\centering

\begin{tikzpicture}[scale=0.5]

\node[draw=black,circle,fill=white] (l1) at (0,0) {1};
\node[draw=black,circle,fill=white] (l2) at (0,-2) {2};
\node[draw=black,circle,fill=white] (l3) at (0,-4) {3};

\node[draw=black,circle,fill=white] (r1) at (4,0) {1};
\node[draw=black,circle,fill=white] (r2) at (4,-2) {2};
\node[draw=black,circle,fill=white] (r3) at (4,-4) {3};

\draw (l1) -- (r1);
\draw (l1) -- (r2);
\draw (l2) -- (r2);
\draw (l2) -- (r3);
\draw (l3) -- (r2);
\draw (l3) -- (r3);
\end{tikzpicture} 
\hspace{1cm}
\begin{tikzpicture}[scale=0.5]

\node[draw=black,circle,fill=white] (l1) at (0,0) {1};
\node[draw=black,circle,fill=white] (l2) at (0,-2) {2};
\node[draw=black,circle,fill=white] (l3) at (0,-4) {3};

\node[draw=black,circle,fill=white] (r1) at (4,0) {1};
\node[draw=black,circle,fill=white] (r2) at (4,-2) {2};
\node[draw=black,circle,fill=white] (r3) at (4,-4) {3};

\draw (l1) -- (r1);
\draw (l1) -- (r2);
\draw (l2) -- (r1);
\draw (l2) -- (r3);
\draw (l3) -- (r1);
\draw (l3) -- (r3);
\end{tikzpicture} 

\caption{Topological structure of the gadgets $(G^1,G^2)$ for bipartite matching.}\label{fig:bm-gadget}
\end{figure}
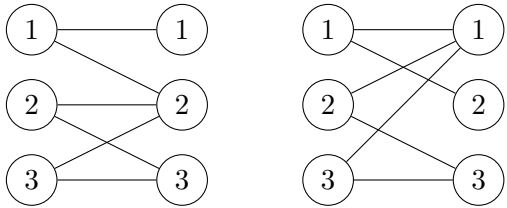

In this case, the (maximum) number $s$ of input items (the number of vertices given) for any of the two gadgets is $3$, $\OPT(G^1)=\OPT(G^2)=3$, and $\BAD(G^1)=\BAD(G^2)=2$. By Corollary~\ref{corollary-form-used}, we can conclude the following:

\begin{theorem}
\label{thm:bipartite_lb}
For Maximum Bipartite Matching and any $\epsilon\in\EPSINT$,
no fixed priority algorithm reading fewer than $(1-H(\epsilon))n/6$ advice bits can
achieve an approximation ratio smaller than $1+\frac{\epsilon}{3-\epsilon}$.
\end{theorem}

Theorem~\ref{thm:bipartite_lb} is related to but also incomparable with the
results by Pena and Borodin~\cite{PenaB16},
showing a $\frac{1}{2}$ asymptotic inapproximation bound for priority
algorithms without advice,
and to the results by D{\"{u}}rr et al.~\cite{Durr2016a} on
inapproximation bounds for online algorithms with advice.

\subsection{Maximum Cut}

Consider the unweighted maximum cut problem in the vertex arrival, vertex adjacency input model.
The goal is to partition vertices into two sets (blocks of the partition) such that the number of edges crossing the two sets is maximized. The partition is specified by an algorithm by assigning $0$ or $1$ to vertices. In addition, we require that $0$ is assigned to vertices belonging to the larger block of the partition. The gadget from Section~\ref{sec:is} (see Figure~\ref{fig:is-gadget}) also works for the maximum cut problem. There is a unique non-trivial maximum cut for that gadget: the cut induced by partitioning vertices into $\{1,2,3\}$ and $\{4,5,6,7,8\}$ for $G^1$ and into $\{6,7,8\}$ and
$\{1,2,3,4,5\}$ for $G^2$. 

Suppose without loss of generality that the highest priority input item
is $(1,\{ 4,5,6,7,8\} )$.
The optimal decision for the first vertex is unique:
For $G^1$, respond $1$, and for $G^2$, respond $0$.

In this case, the maximum number $s$ of input items for a gadget is $8$, $\OPT(G^1)=\OPT(G^2)=15$, and $\BAD(G^1)=\BAD(G^2)=14$. By Corollary~\ref{corollary-form-used}, we can conclude the following: 

\begin{theorem}
For Maximum Cut and any $\epsilon\in\EPSINT$,
no fixed priority algorithm reading fewer than $(1-H(\epsilon))n/16$ advice bits can
achieve an approximation ratio smaller than $1+\frac{\epsilon}{15-\epsilon}$.
\end{theorem}


\subsection{Maximum Satisfiability}

We consider the MAX-SAT problem (and, in fact, MAX-3-SAT) in the following input model. An input item $(x, S^+,S^-)$ consists of a variable name $x$, a set $S^+$ of \emph{clause information tuples} for those clauses in which $x$ appears positively, and a set $S^-$ of clause information tuples for those clauses where the variable $x$ appears negatively. The clause information tuples for a particular clause contain the name of the clause, the total number of literals in that clause, and the names of the other variables in the clause, but no information regarding whether those other variables are negated or not. This corresponds to Model 2 in~\cite{PenaB16}.
A gadget is then a set of input items defining a consistent CNF-SAT formula.
Thus, for every clause information
tuple $(C,\ell, V)$ for a variable $x$ with
$V=\{ x_{i_1},x_{x_2},\ldots,x_{i_r}\}$, we have that
$\ell=r+1$ (since the variable itself is in the clause along with $r$
other literals), and for each $x_{i_j}$, the variable $x$ occurs in an
information tuple associated with $x_{i_j}$, along with the same clause
name $C$ and the same length $\ell$.
Two input items are isomorphic if the are the same up to
renaming of the variables.
The goal is to satisfy the maximum number of clauses. Consider the following pair of instances (gadgets):

\[ G^1 = C_1 \wedge C_2 \wedge C_3 \wedge C_4 \wedge C_5 \wedge C_6 \wedge C_7 
\wedge C_8, \]
where
$$\begin{array}{ll}
C_1= (x_1 \vee x_2 \vee x_3) & C_2= (x_1 \vee \lnot x_2 \vee \lnot x_3) \\[1ex]
C_3= (x_1 \vee \lnot x_2 \vee x_3) & C_4= (x_1 \vee x_2 \vee \lnot x_3) \\[1ex]
C_5= (\lnot x_1 \vee x_2 \vee x_3) & C_6= (\lnot x_1 \vee x_2 \vee x_3) \\[1ex]
C_7= (\lnot x_1 \vee \lnot x_2 \vee \lnot x_3) & C_8= (\lnot x_1 \vee \lnot x_2 \vee \lnot x_3) 
\end{array}$$

There are only 3 variables, each appearing in every clause.
In addition, each variable occurs
positively in four clauses and negatively in four others.

When restricting the clauses $C_1$ through $C_4$ to just the variables
$x_2$ and $x_3$, the result is all possible clauses over $x_2$ and $x_3$.
Therefore, no 
truth assignment for $x_2$ and $x_3$ can satisfy all four clauses,
unless $x_1$ is set to True.
To satisfy $C_5$ through $C_8$, we can set $x_2$ to True and $x_3$ to False. 
Thus, every maximum assignment has $x_1$ set to True.

Consider
\[ G^2 = C_1 \wedge C_2 \wedge C_3 \wedge C_4 \wedge C_5 \wedge C_6 \wedge C_7 
\wedge C_8, \]
where
$$\begin{array}{ll}
C_1= (\lnot x_1 \vee x_2 \vee x_3) & C_2= (\lnot x_1 \vee \lnot x_2 \vee \lnot x_3) \\[1ex]
C_3= (\lnot x_1 \vee \lnot x_2 \vee x_3) & C_4= (\lnot x_1 \vee x_2 \vee \lnot x_3) \\[1ex]
C_5= (x_1 \vee x_2 \vee x_3) & C_6= (x_1 \vee x_2 \vee x_3) \\[1ex]
C_7= (x_1 \vee \lnot x_2 \vee \lnot x_3) & C_8= (x_1 \vee \lnot x_2 \vee \lnot x_3) 
\end{array}$$

The universe of inputs for these gadgets consists of all input items 
$(x,S^+,S^-)$, where $x\in \{ x_1,x_2,x_3 \}$, and each of $S^+$ and $S^-$ 
contain  four distinct clause information tuples with clause names in the
set $\{ C_1,C_2,C_3,C_4,C_5,C_6,C_7,C_8\}$, lengths equal to $3$, and
variable sets containing the other two variables not equal to $x$. All
eight clause names will appear in every input item.

Suppose without loss of generality that the highest priority input among
all of these possibilities is
\[
(x_1, \{ (C_1,3,\{ x_2,x_3\}),(C_2,3,\{ x_2,x_3\}),(C_3,3,\{ x_2,x_3\}),(C_4,3,\{ x_2,x_3\})\},\]
\[\{ (C_5,3,\{ x_2,x_3\}),(C_6,3,\{ x_2,x_3\}), (C_7,3,\{ x_2,x_3\}),(C_8,3,\{ x_2,x_3\}) \}).
\]
Note that the optimal decision for $x_1$ is unique for each of these gadgets
and is ``True'' for $G^1$ and ``False'' for $G^2$.

In this case, the maximum number $s$ of input items for a gadget is $3$, $\OPT(G^1)=\OPT(G^2)=8$, and $\BAD(G^1)=\BAD(G^2)=7$. By Corollary~\ref{corollary-form-used}, we can conclude the following:

\begin{theorem}
\label{thm:max-sat_lb}
For Maximum $3$-Satisfiability and any $\epsilon\in\EPSINT$,
no fixed priority algorithm reading fewer than $(1-H(\epsilon))n/6$ advice bits can
achieve an approximation ratio smaller than $1+\frac{\epsilon}{8-\epsilon}$.
\end{theorem}

Note that the gadget pair used in the proof above has repeated clauses.
We believe it is possible to prove a similar result without repeated clauses at the expense of a more complicated gadget.

Theorem~\ref{thm:max-sat_lb} is related to but incomparable with the Poloczek \cite{Poloczek11} Maximum Satisfiability inapproximation result
 for adaptive priority algorithms (without advice) that, as in our theorem, uses their input Model~2. 

\subsection{A Job Scheduling Problem}

In this section, we consider job scheduling on a single machine of unit time jobs with precedence constraints. In this problem, we are given a set of jobs with precedence constraints specifying, for example, that if job $J_1$ and job $J_2$ are scheduled, then $J_1$ has to precede job $J_2$. The precedence constraints are not necessarily compatible, i.e., there could be a cyclic set of constraints. We are interested in scheduling a maximum number of jobs that are compatible. We can think of the precedence constraints as specifying a directed graph, in which case it is called the maximum induced directed acyclic subgraph problem. This problem is the complement of the minimum feedback vertex set problem -- one of Karp's original NP-complete problems~\cite{Karp72}.
Inapproximation bounds were proven by Lund and Yannakakis in~\cite{LY93}.
The schedule can be obtained from such a subgraph by ordering the jobs topologically and scheduling them one after another in that order. Thus, the input items are of the form $(J, S^+, S^-)$, where $J$ is the name of a job, $S^+$ is the set of jobs such that if they were scheduled together with $J$ they would have to be scheduled before $J$, and $S^-$ is the set of jobs such that if they were scheduled together with $J$ they would have to be scheduled after $J$. Using graph terminology, $S^+$ consists of all incoming neighbors of $J$ and $S^-$ consists of all outgoing neighbors of $J$. An input item describes a subgraph consisting of a distinguished vertex together with all of its predecessors and successors and all edges connecting to or from the distinguished vertex. Two input items are considered isomorphic if they are isomorphic as graphs. This implies in particular that they have the same in- and out-degrees. Figure~\ref{fig:js-gadget} shows a topological gadget such that every optimal solution contains Job~0 and excludes Job~8, and it consists only of isomorphic items (each vertex has in-degree $2$, out-degree $2$, and $4$ different neighbors in all).

\begin{figure}[h]
\centering
\tikzstyle{vertex}=[circle, draw, minimum size=5mm] 
\newcommand{\vertex}{\node[vertex]}
\usetikzlibrary{arrows.meta}
\tikzset{%
  tipA/.tip={Triangle[open,angle=75:3pt]},
  tipB/.tip={Triangle[angle=75:5pt]},
}
\begin{tikzpicture}[scale=1.0] 
\vertex (p0) at (0,2) {0};
\vertex (p1) at (2,2) {1};
\vertex (p2) at (4,2) {2};
\vertex (p3) at (6,2) {3};
\vertex (p4) at (0,0) {4};
\vertex (p5) at (2,0) {5};
\vertex (p6) at (4,0) {6};
\vertex (p7) at (6,0) {7};
\vertex (p8) at (1,1) {8};
\draw[-tipB] (p0) -- (p8);
\draw[-tipB] (p0) to[out=45,in=135,distance=7mm] (p2);
\draw[-tipB] (p1) -- (p0);
\draw[-tipB] (p1) -- (p2);
\draw[-tipB] (p2) -- (p3);
\draw[-tipB] (p2) -- (p7);
\draw[-tipB] (p3) to[out=90,in=90,distance=10mm] (p0);
\draw[-tipB] (p3) to[out=135,in=45,distance=7mm] (p1);
\draw[-tipB] (p4) to[out=-45,in=-135,distance=7mm] (p6);
\draw[-tipB] (p4) -- (p8);
\draw[-tipB] (p5) -- (p4);
\draw[-tipB] (p5) -- (p6);
\draw[-tipB] (p6) -- (p3);
\draw[-tipB] (p6) -- (p7);
\draw[-tipB] (p7) to[out=-90,in=-90,distance=10mm] (p4);
\draw[-tipB] (p7) to[out=-135,in=-45,distance=7mm] (p5);
\draw[-tipB] (p8) -- (p1);
\draw[-tipB] (p8) -- (p5);
\end{tikzpicture}
\caption{Topological structure of a gadget for job scheduling of unit time jobs with precedence constraints.}\label{fig:js-gadget}
\end{figure}
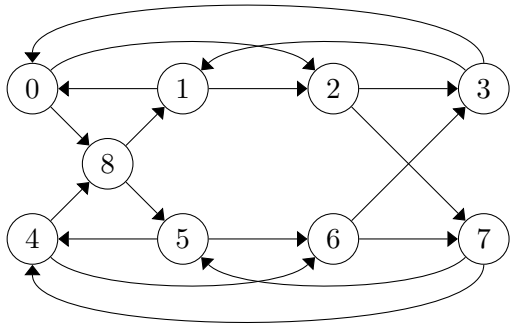

In this case, the maximum number $s$ of input items for a gadget is $9$, $\OPT(G^1)=\OPT(G^2)=6$ (for instance, schedule Jobs $1,0,2,5,4,6$), and $\BAD(G^1)=\BAD(G^2)=5$. By Corollary~\ref{corollary-form-used}, we can conclude the following:

\begin{theorem}
For Job Scheduling of Unit Time Jobs with Precedence Constraints and any $\epsilon\in\EPSINT$,
no fixed priority algorithm reading fewer than $(1-H(\epsilon))n/18$ advice bits can
achieve an approximation ratio smaller than $1+\frac{\epsilon}{6-\epsilon}$.
\end{theorem}

\subsection{Vertex Cover}
\label{VC}
Consider the minimum vertex cover problem in the vertex arrival, vertex adjacency input model.

We use the construction from~\cite{BorodinBLM2010} to obtain two pairs of
gadgets, one if the highest priority input item has degree~$2$ and the other if it
has degree~$3$. For each input $x$ to Pair Matching, the universe of input items contains
names of seven vertices, and for each of the vertices all possibilities for both
degrees two and three.

First note that both graphs in Fig.~\ref{Graph3} have vertex
covers of size~$3$.
\begin{figure}[htp]
\centering
\begin{tikzpicture}[scale=0.5]

\node[draw=black,circle,fill=white] (1) at (0,-8) {4};
\node[draw=black,circle,fill=white] (2) at (-3,-6) {3};
\node[draw=black,circle,fill=white] (3) at (0,-6) {7};
\node[draw=black,circle,fill=white] (4) at (3,-6) {5};
\node[draw=black,circle,fill=white] (5) at (-3,-4) {2};
\node[draw=black,circle,fill=white] (6) at (3,-4) {6};
\node[draw=black,circle,fill=white] (7) at (0,-2) {1};

\draw (1) -- (2);
\draw (1) -- (4);
\draw (2) -- (3);
\draw (3) -- (4);
\draw (2) -- (5);
\draw (5) -- (7);
\draw (4) -- (6);
\draw (6) -- (7);
\end{tikzpicture}
\hspace{1cm}
\begin{tikzpicture}[scale=0.5]

\node[draw=black,circle,fill=white] (1) at (0,-2) {4};
\node[draw=black,circle,fill=white] (2) at (0,-5) {3};
\node[draw=black,circle,fill=white] (3) at (3,-2) {7};
\node[draw=black,circle,fill=white] (4) at (3,-8) {5};
\node[draw=black,circle,fill=white] (5) at (0,-8) {2};
\node[draw=black,circle,fill=white] (6) at (3,-5) {6};
\node[draw=black,circle,fill=white] (7) at (-3,-5) {1};

\draw (1) -- (7);
\draw (5) -- (7);
\draw (2) -- (7);
\draw (1) -- (3);
\draw (2) -- (3);
\draw (2) -- (6);
\draw (5) -- (6);
\draw (5) -- (4);
\draw (4) to[out=45,in=45,distance=55mm] (1);
\end{tikzpicture}
\caption{Graph 1 to the left and Graph 2 to the right.}
\label{Graph3}
\end{figure}
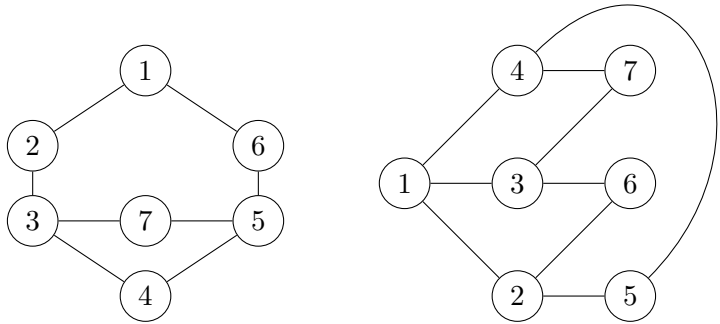

However, in order to obtain a vertex cover of size~$3$, it is necessary
to accept vertex~$1$ in Graph~1 and reject vertex~$2$ in Graph~1. Thus,
the gadget pair for vertices of degree~$2$ consists of two copies of
Graph~1, where the highest priority vertex is vertex~$1$ in the first
gadget and vertex~$2$ in the second. 

Similarly, in order to obtain a vertex cover of size~$3$, it is necessary
to accept vertex~$3$ in Graph~1 and reject vertex~$1$ in Graph~2. Thus,
the gadget pair for vertices of degree~$3$ consists of Graph~1, where the
highest priority vertex is vertex~$3$, and Graph~2, where the highest priority
vertex is vertex~$1$.

The highest priority vertex must have one of these two degrees, so the
reduction can continue with the correct gadget pair for that degree.

For either gadget pair, the maximum number $s$ of input items for a gadget is 
$7$, $\OPT(G^1)=\OPT(G^2)=3$, and $\BAD(G^1)=\BAD(G^2)=4$. 
By  Corollary~\ref{corollary-form-used}, we can conclude the following:

\begin{theorem}
For Minimum Vertex Cover and any $\epsilon\in\EPSINT$,
no fixed priority algorithm reading fewer than $(1-H(\epsilon))n/14$ advice bits can
achieve an approximation ratio smaller than $1+\frac{\epsilon}{3}$.
\end{theorem}

Below we show a weaker result using a regular graph, so all input items are isomorphic.

Consider the topological structure of a gadget in Figure~\ref{fig:vc-gadget}. It is a 4-regular graph on 8 vertices. This graph has a unique, non-trivial minimum vertex cover $\{2,3,4,6,8\}$ (we have verified by enumeration). Note that this is very similar to the case for Independent Set, in that an isomorphic copy of the same graph can be used for the other gadget in the pair. Then, assuming that $(2,\{ 1,3,4,7 \} )$ is the first input item, accepting the vertex can
lead to the unique optimum vertex cover in the gadget depicted, and renaming
the vertex to one different from $\{2,3,4,6,8\}$ and rejecting the vertex
can lead to the unique optimum vertex cover in a second gadget.

\begin{figure}[h]
\centering

\begin{tikzpicture}[scale=0.5]

\node[draw=black,circle,fill=white] (1) at (-4,-4) {1};
\node[draw=black,circle,fill=white] (2) at (-2,-2) {2};
\node[draw=black,circle,fill=white] (3) at (0,0) {3};
\node[draw=black,circle,fill=white] (4) at (2,-2) {4};
\node[draw=black,circle,fill=white] (5) at (4,-4) {5};

\node[draw=black,circle,fill=white] (6) at (2,-6) {6};
\node[draw=black,circle,fill=white] (7) at (0,-8) {7};
\node[draw=black,circle,fill=white] (8) at (-2,-6) {8};

\draw (1) -- (2);
\draw (2) -- (3);
\draw (3) -- (4);
\draw (4) -- (5);
\draw (5) -- (6);
\draw (6) -- (7);
\draw (7) -- (8);
\draw (8) -- (1);

\draw (2) -- (4);
\draw (4) -- (7);
\draw (7) -- (2);

\draw (1) edge[bend left] (3);
\draw (3) edge[bend left] (5);
\draw (5) -- (8);
\draw (8) -- (6);
\draw (6) -- (1);

\end{tikzpicture} 

\caption{Topological structure of a gadget for vertex cover.}\label{fig:vc-gadget}
\end{figure}
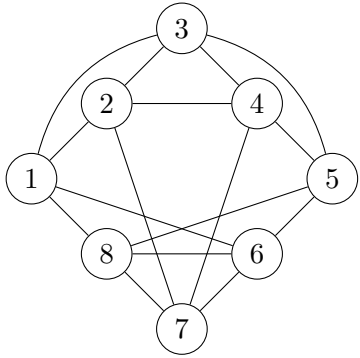

In this case, the maximum number $s$ of input items for a gadget is $8$, $\OPT(G^1)=\OPT(G^2)=5$, and $\BAD(G^1)=\BAD(G^2)=6$. By  Corollary~\ref{corollary-form-used}, we can conclude the following: 

\begin{theorem}
For Minimum Vertex Cover and any $\epsilon\in\EPSINT$,
no fixed priority algorithm reading fewer than $(1-H(\epsilon))n/16$ advice bits can
achieve an approximation ratio smaller than $1+\frac{\epsilon}{5}$.
\end{theorem}


\section{Concluding Remarks}
\label{sec:conclusion}
We have developed a general framework for showing linear lower bounds on the number of advice bits required to get a constant approximation ratio for fixed priority algorithms with advice. The framework relies on reductions from the Pair Matching problem --- an analogue of the Binary String Guessing problem from the online world, resistant to universe orderings. Many problems remain open:
\begin{itemize}

\item Can our framework (or a modification of it) show non-constant inapproximation results with large advice, for example, for independent set?


\item In vertex coloring, any decision for the first item can be completed to an optimal solution. Can our framework be modified to handle such problems? For example, see an argument for the makespan problem in \cite{Regev02}.

\item An interesting goal is to study the ``structural complexity'' of online and priority algorithms. Can one define analogues of classes such as NP, NP-Complete, $\sharp$P, etc. for online/priority problems? If so, are complete problems for these classes natural? 

\end{itemize}

\bigskip
{\bf Acknowledgements.} Part of the work was done when the second and third authors were visiting the University of Toronto, the first author was visiting Toyota Technological Institute at Chicago, and the fourth author was a postdoc at the University of Toronto.


\bibliography{refs} 
\bibliographystyle{plain}





\end{document}